\documentclass[12pt,reqno,fleqn]{amsart}
\usepackage{latexsym}
\usepackage{amssymb}
\usepackage{amsxtra}
\usepackage{amsmath}
\usepackage{amsfonts}

\usepackage{graphicx}
\usepackage{epstopdf}

\usepackage{amscd}
\usepackage{graphics}
\newcommand{\cleqn}{\setcounter{equation}{0}}
\newcommand{\clth}{\setcounter{theorem}{0}}
\newcommand {\sectionnew}[1]{\section{#1}\cleqn\clth}
\newtheorem{theorem}{Theorem}[section]
\newtheorem{lemma}[theorem]{Lemma}
\newtheorem{proposition}[theorem]{Proposition}
\newtheorem{corollary}[theorem]{Corollary}
\newtheorem{definition}[theorem]{Definition}
\newcommand{\C}{\mathcal{C}}
\renewcommand{\P}{\mathcal{P}}
\newcommand{\W}{\mathcal{W}}
\renewcommand{\L}{\mathcal{L}}

\newcommand{\Z}{\mathbb{Z}}
\newcommand{\N}{\mathbb{N}}

\newcommand{\B}{\mathcal{B}}

\renewcommand{\t}{\mathbf{t}}

\makeatletter
    
    \newcommand{\Rmnum}[1]{\expandafter\@slowromancap\romannumeral #1@}
  \makeatother
\def\res{\mathop{\rm Res}\nolimits}
\def\({\left(}
\def\){\right)}
\def\[{\left[}
\def\]{\right]}
\def\d{\partial}
\def\ga{\alpha}

\begin{document}

\title[Extended  bigraded Toda Hierarchy]{Tau function and  Hirota bilinear equations for the Extended  bigraded Toda Hierarchy}
\author{Chuanzhong Li\dag\ddag
, Jingsong He $^*$ \dag \S, Ke Wu\ddag, Yi Cheng \dag }
\dedicatory {  \dag Department of Mathematics, USTC, Hefei, 230026 Anhui, P.\ R.\ China\\
 \S Department of Mathematics, NBU, Ningbo, 315211 Zhejiang, P.\ R.\ China\\
 \ddag Department of Mathematics, CNU, 100037 Beijing, P.\ R.\ China }

\thanks{$^*$ Corresponding author:\ hejingsong@nbu.edu.cn,jshe@ustc.edu.cn,}
\texttt{}

\date{}

%%%%%%%%%%%%%%%%%%%%%%%%%%%%%%%%%%%%%%%%%%%%%%%%
\begin{abstract}
In this paper we generalize the Sato theory to the extended bigraded
Toda hierarchy (EBTH). We revise the definition of the Lax
equations,
 give the Sato equations, wave operators, Hirota bilinear identities (HBI) and show the existence
of $tau$ function $\tau(t)$. Meanwhile we prove the validity of  its
Fay-like identities and  Hirota bilinear equations (HBEs) in terms
of
 vertex operators whose coefficients  take values in the algebra of differential
operators. In contrast with HBEs of the usual integrable system, the
current HBEs are equations of product of operators involving
$e^{\partial_x}$ and $\tau(t)$.

\end{abstract}

%%%%%%%%%%%%%%%%%%%%%%%%%%%%%%%%%%%%%%%%%%%%%%%%

\maketitle
\allowdisplaybreaks
 \setcounter{section}{0}

\sectionnew{Introduction}
The Toda lattice equations is a set of nonlinear evolutionary differential-difference equations
introduced by Toda (\cite{Toda}, \cite{Toda book}) describing an infinite system of masses on a line that interact through
an exponential force which is used to explain nonergodic character in the well-known Fermi-Pasta-Ulam
paradox. It was soon realized that this lattice equations is a completely integrable system,
i.e. admits infinite conserved quantities and exact analytic solutions. It has important applications in many different fields such
 as classical and quantum fields  theory. For our best knowledge, there are at least three important extensions of
 Toda lattice equation. The first one is the Toda hierarchy \cite{UT}, which is in fact a two-dimensional
extended hierarchy through infinite-dimensional matrix inspired by the Sato theory\cite{Sato Fay identity}. Recently, considering application
 to 2D topological fields  theory  and  the theory of Gromov-Witten invariants (\cite{Z}, \cite{D witten},
 \cite{witten},\cite{dubrovin}) of Toda lattice hierarchy, one replaced the discrete variables with continuous
one.  After continuous `` interpolation'' \cite{CDZ} to the whole Toda lattice hierarchy, it was found the
flow of spatial translations was missing. In order to get a complete family of flows \cite{DZ}, the interpolated
Toda lattice hierarchy was extended into the so-called extended Toda hierarchy(ETH) \cite{CDZ}, which is the second
extension  of the Toda lattice equations. It was  firstly conjectured and then shown (\cite{Z}, \cite{DZ}, \cite{Ge}) that
the extended Toda hierarchy is the hierarchy describing the Gromov-Witten invariants of $CP^1$ by matrix models
 \cite{matrix model} which describe in the large $N$ limit of the $CP^1$ topological sigma model.
The HBEs of the ETH are given by Milanov 's work\cite{M}. The third
extension of Toda lattice equations is the extended bigraded Toda hierarchy(EBTH),  which are discovered independently two
times from different concerns. The dispersionless version of extended bigraded Toda hierarchy was firstly
introduced by S. Aoyama, Y. Kodama in \cite{Kodama CMP}. In the dispersionless limit, the EBTH can be obtained from the
dispersionless KP hierarchy.  More recently, the extended bigraded Toda hierarchy was re-introduced by Gudio
Carlet \cite{C} who hoped that EBTH might also be relevant for some applications in 2D topological fields
theory and in the theory of Gromov-Witten invariants. Specifically, Carlet \cite{C} generalized the Toda
lattice hierarchy by considering $N+M$ dependent variables and used them to provide a Lax pair definition of
the extended bigraded Toda hierarchy. On the base of \cite{C}, Todor. E. Milanov and Hsian-Hua Tseng \cite{TH}
described conjecturally one kind of Hirota bilinear equations (HBEs) which was proved to  govern the
Gromov-Witten theory of orbiford $c_{km}$. This naturally inspires us to consider the Sato theory of
EBTH because the  HBEs are the core knowledge of the integrable systems.

Sato and Sato proved one kind of algebraic identity in theorem of \cite{Sato Fay identity} about tau function of KP
hierarchy  which is now called the Fay-identity\cite{Fay book}. There are some important extensions, such as differential
Fay-identity and its implication \cite{Adler and P van Moerbeke}. Furthermore,
Takasaki and Takebe derived the differential Fay identity from the Hirota bilinear identity(HBI) of KP hierarchy
and showed that the differential Fay identity is equivalent to KP hierarchy in appendix of \cite{Takasaki Takebe}.
In \cite{Takasaki fay}, it shows that differential (or difference, for the Toda hierarchy)
Fay identities are  generating functional expression of the full set of auxiliary linear equations and  equivalent to the integrable hierarchies themselves.
Lee-Peng Teo derived the Fay-like identities of tau function
for the Toda lattice hierarchy from the HBI and  prove that the Fay-like identities are equivalent to the
hierarchy \cite{Fay-like}.

 So the purpose of this paper is to  establish Sato formulation of EBTH including its Lax equations, Sato
equations, wave operators, HBIs, tau-functions, Fay-like identities and HBEs. An important feature of the current HBEs
is that they are not differential  equations of functions as the case of usual integrable systems. Actually, HBEs of
the EBTH are equations of product of operators involving $e^{\partial_x}$ and $\tau(t)$. Here $\tau$ function is
regarded  as a zero order operator. In other words, the coefficients of vertex operators with a form of
$e^{\partial_x+ \sum_{\alpha,n}\theta_{\alpha,n}\partial_{t_{\alpha,n}}+ \sum_{\beta,m}\theta_{\beta,m}\partial_{t_{\beta,m}}}$ in the HBEs,
are not scalar-valued  but take values in the algebra of differential operators, i.e. $\{\d_x,x\}$.
This  subtle point of HBEs can be found in Milanov's work \cite{M}. Our work is an highly nontrivial extension of
the results in \cite{M} which is about ETH. We would like to stress that  the current HBEs in section 5 are
different from the one in \cite{TH}.

The paper is organized as follows. In Section 2, we redefine  the Lax equations using  the
roots and the logarithms of the Lax operator $\L$ and give  Zakharov-Shabat equation and Sato equations for the
EBTH. By using the wave operators and their symbols, some bilinear identities are given in Section 3. In Section 4
we define the tau-function of EBTH and  prove its existence, moreover we give some Fay-like identities from HBIs
under some special cases. In Section 5  we give the HBEs of EBTH in the form of tau function and vertex
operators, meanwhile we prove its validity with the help of HBI. In Section 6  we give the HBEs of bigraded Toda
hierarchy (BTH) as a corollary. Section 7 is devoted to conclusions and discussions.

\sectionnew{ The EBTH }
We describe the Lax form of the EBTH following \cite{C}. Introduce
firstly the lax operator
\begin{equation}\L=\Lambda^{N}+u_{N-1}\Lambda^{N-1}+\dots + u_{-M}
\Lambda^{-M}
\end{equation}
which can be expressed in the following two different ways
  \begin{eqnarray}
  \label{dressing}\L=\P_L\Lambda^N\P_L^{-1} = \P_R
  \Lambda^{-M}\P_R^{-1}.
  \end{eqnarray}
  Here, $N,M \geq1$ are two fixed positive integers and $u_{-M}$ is a
non-vanishing function. The variables $u_j$ are functions of the
spatial variable $x$ and the
  shift operator $\Lambda$ acts on a function $a(x)$ by $\Lambda a(x) = a(x + \epsilon
  )$,
i.e. $\Lambda$ is equivalent to  $e^{\epsilon\partial_{x}}$ where
the spacing unit $``\epsilon"$ is called string coupling constant.
The operators $\P_L$ and $ \P_R$ have the following forms
\begin{eqnarray}
&& \P_L=1+w_1\Lambda^{-1}+w_2\Lambda^{-2}+\ldots,
\label{dressP}\\
&& \P_R=\tilde{w_0}+\tilde{w_1}\Lambda+\tilde{w_2}\Lambda^2+ \ldots,
\label{dressQ} \end{eqnarray} where $\tilde{w_0}$ is  not zero.
The inverse operators of $P_L$ and $P_R$ are given by
\begin{eqnarray}
\P_L^{-1} &=& 1+\Lambda^{-1}w_1'+\Lambda^{-2}w_2'+\ldots, \label{dressPinv}\\
\P_R^{-1}&=&\tilde w'_0 + \Lambda\tilde w'_1 +  \Lambda^2\tilde w'_2
+\ldots. \label{dressQinv}
\end{eqnarray}
Note that the operator $\Lambda^i$ are fixed at the left side of coefficients in inverse operators.
The uniqueness is up to multiplying $\P_L$ and $\P_R$ from the right
 by operators  in the form  $1+
a_1\Lambda^{-1}+a_2\Lambda^{-2}+...$ and $\tilde{a}_0 +
\tilde{a}_1\Lambda +\tilde{a}_2\Lambda^2+\ldots$ respectively whose
coefficients are independent of $x$. From the first identity of
eq.\eqref{dressing}, we can easily get the relation of $u_i$ and
$w_j$ as following
\begin{eqnarray}
  u_{N-1}&=&w_1(x)-w_1(x+N\epsilon), \\
  u_{N-2}&=&w_2(x)-w_2(x+N\epsilon)-(w_1(x)-w_1(x+N\epsilon))w_1(x+(N-1)\epsilon), \\\notag
   u_{N-3}&=&w_3(x)-w_3(x+N\epsilon)- [w_2(x)-w_2(x+N\epsilon)
  -(w_1(x)-w_1(x+N\epsilon))w_1(x+(N-1)\epsilon)]\\
  &&w_1(x+(N-2)\epsilon)-(w_1(x)-w_1(x+N\epsilon))w_2(x+(N-1)\epsilon), \\
  \notag \cdots &\cdots &\cdots
\end{eqnarray}\\
Moreover, by using the second identity of eq.\eqref{dressing} and
the non-vanishing character of $\tilde{w_0}$, we can also easily get
the relation of $u_i$ and $\tilde w_j$ formally as following
\begin{eqnarray*}
  u_{-M}&=&\frac{\tilde w_0(x)}{\tilde w_0(x-M\epsilon)}, \\
  u_{-M+1}&=&\frac{\tilde w_1(x)-\frac{\tilde w_0(x)}{\tilde w_0(x-M\epsilon)}\tilde w_1(x-M\epsilon)}{\tilde w_0(x-(M-1)\epsilon)}, \\\notag
   u_{-M+2}&=&\frac{\tilde w_2(x)-\frac{\tilde w_0(x)}{\tilde w_0(x-M\epsilon)}\tilde w_2(x-M\epsilon)-\frac{\tilde w_1(x)-\frac{\tilde w_0(x)}
   {\tilde w_0(x-M\epsilon)}\tilde w_1(x-M\epsilon)}{\tilde w_0(x-(M-1)\epsilon)}\tilde w_1(x-(M-1)\epsilon)}{\tilde w_0(x-(M-2)\epsilon)},\\
  &&
\\
  \notag \cdots &\cdots &\cdots\\\notag
 u_{N-1}&=&\frac{\tilde w_{M+N-1}-u_{-M}\tilde w_{M+N-1}(x-M\epsilon)
   -\dots-u_{N-2}\tilde w_1(x+(N-2)\epsilon)}{\tilde w_0(x+(N-1)\epsilon)},\\\notag
 u_N=  1&=&\frac{\tilde w_{M+N}-u_{-M}\tilde w_{M+N}(x-M\epsilon)
   -\dots-u_{N-1}\tilde w_1(x+(N-1)\epsilon)}{\tilde
   w_0(x+N\epsilon)}.
\end{eqnarray*}\\
To write out  explicitly the Lax equations  of EBTH,
  fractional powers $\L^{\frac1N}$ and
$\L^{\frac1M}$ was defined by
\begin{equation}
  \notag
  \L^{\frac1N} = \Lambda+ \sum_{k\leq 0} a_k \Lambda^k , \qquad \L^{\frac1M} = \sum_{k \geq -1} b_k
  \Lambda^k,
\end{equation}
with the relations
\begin{equation}
  \notag
  (\L^{\frac1N} )^N = (\L^{\frac1M} )^M = \L.
\end{equation}
It was  stressed that $\L^{\frac1N}$ and $\L^{\frac1M}$ are two
different operators even if $N=M(N, M\geq 2)$ in \cite{C} due to two different dressing operators. They can also be
expressed as following
\begin{equation}
  \notag
  \L^{\frac1N} = \P _{L}\Lambda\P_{L}^{-1}, \qquad \L^{\frac1M} = \P_{R}\Lambda^{-1} \P_{R}^{ -1}.
\end{equation}
Moreover, as \cite{C} it is necessary  to define the  following two
logarithms
\begin{subequations}  \notag
 \begin{align}
  &\log_+ \L = \P_{L} N \epsilon \partial \P_{L}^{-1} = N \epsilon \partial - N \epsilon \P_{Lx} \P_{L}^{-1}
  = N \epsilon \partial + 2 N \sum_{k > 0} W_{-k}(x) \Lambda^{-k}, \\
  &\log_- \L = - \P_{R}M \epsilon\partial\P_{R}^{-1} = - M \epsilon \partial + M \epsilon \P_{Rx} \P_{R}^{-1}
   = -M \epsilon \partial + 2 M \sum_{k \geq 0} W_k(x) \Lambda^k,
\end{align}
\end{subequations}

where $\partial = \frac{d}{dx}$ and $\P_{Lx}, \P_{Rx}$ are differentiating $\P_{L}, \P_{R}$ respectively with respect to $x$. Now define

\begin{equation}
   \notag
  \log\L = \frac1{2N} \log_+\L + \frac1{2M} \log_- \L = \sum_{k \in \Z} W_k
  \Lambda^k.
\end{equation}

Given any difference operator $A= \sum_k A_k \Lambda^k$, the
positive and negative projections are given by $A_+ = \sum_{k\geq0}
A_k \Lambda^k$ and $A_- = \sum_{k<0} A_k \Lambda^k$. Similar to \cite{C}, we give the following definition.
\begin{definition} \label{deflax}
The Lax equations of extended bigraded Toda hierarchy is given by
\begin{equation}
  \label{edef}
\frac{\partial \L}{\partial t_{\alpha, n}} = [ A_{\alpha,n} ,\L ]
\end{equation}
for $\alpha = N,N-1,N-2, \dots, -M$ and $n \geq 0$. Here operators
$A_{\alpha ,n}$ are defined by
\begin{subequations}
\label{Adef}
\begin{align}
  &A_{\alpha,n} = \frac{\Gamma (2- \frac{\alpha-1}{N} )}{  \epsilon \Gamma(n+2 -\frac{\alpha-1}{N} ) } ( \L^{n+1-\frac{\alpha-1}N })_+ \quad \text{for} \quad \alpha = N,N-1, \dots, 1,\\
  &A_{\alpha,n} = \frac{-\Gamma (2+\frac{\alpha}{M} )}{  \epsilon \Gamma(n+2 +\frac{\alpha}{M} ) } ( \L^{n+1+\frac{\alpha}M })_- \quad \text{for} \quad \alpha = 0,-1, \dots, -M+1, \\
  &A_{-M,n} = \frac{2}{\epsilon n!} [ \L^n (\log \L - \frac12 ( \frac1M + \frac1N) \C_n ) ]_+,
\end{align}
\end{subequations}
and the constants $\C_n$ are defined by
\begin{equation}
  \label{b24}
  \C_n = \sum_{k=1}^n \frac1k ,  \C_0=0 .
\end{equation}
\end{definition}
The only difference of this definition from \cite{C} is that we add
the hierarchy when $\alpha=1$ to the hierarchies in the definition
of \cite{C}. That hierarchy is in fact the Toda hierarchy which is
also the hierarchy when $\alpha=0$. We do this because it is necessary
to introduce such an additional group of equations for proving the
existence of tau function.

 Particularly for $N=M=1$ this hierarchy
coincides with the extended Toda hierarchy introduced in \cite{CDZ}.
 If we consider $\L^{\frac1N}$ and
$\L^{\frac1M}$ are two completely independent operators, the EBTH
will imply well-known 2-dimensional Toda hierarchy. We can consider
the  EBTH as a kind of extended constrained  2-dimensional Toda
hierarchy with constraint $(\L^{\frac1N} )^N = (\L^{\frac1M} )^M$.
\\
For the convenience to lead to the Sato equation, we   define the
following operators which are also similar to \cite{C}:

\begin{equation}
  B_{\alpha , n} :=
\begin{cases}
  \frac{\Gamma ( 2- \frac{\alpha-1}{N} )}{\epsilon\Gamma (n+2 - \frac{\alpha-1}{N} )} \L^{n+1-\frac{\alpha-1}{N}}, &\alpha=N\dots 1,\\
  \frac{\Gamma ( 2 + \frac{\alpha}{M} )}{\epsilon\Gamma (n+2 + \frac{\alpha}{M} )}  \L^{n+1+\frac{\alpha}{M}}, &\alpha = 0\dots -M+1,\\
 \frac{2}{\epsilon n!} [ \L^n( \log \L - \frac{1}{2}( \frac{1}{M} + \frac{1}{N} ) c_n) ], & \alpha = -M.
  \end{cases}
\end{equation}
Then  the following lemma can be got  \cite{C}.
\begin{lemma} \label{lemD} The following equations hold
\begin{align}\label{dd}
&\partial_{\alpha,p}\L^n=[ A_{\alpha,p}, \L^n],\\
  \label{d6}
   & (\L^{\frac1N})_{t_{\alpha,p}} = [ - (B_{\alpha,p})_-, \L^{\frac1N} ], \\
  \label{d6i} &(\L^{\frac1M})_{t_{\alpha,p}} = [ (B_{\alpha,p})_+, \L^{\frac1M} ], \\
   \label{d6ii}&(\log_+ \L)_{t_{\alpha,p}} = [ -(B_{\alpha,p})_-, \log_+\L ], \\
   \label{d6iii}&(\log_- \L)_{t_{\alpha,p}} = [(B_{\alpha,p})_+ ,\log_- \L
   ],
\end{align}
and combine the last two equations into
\begin{equation}
  \label{d7}
   (\log \L)_{t_{\alpha,p}} = [ -(B_{\alpha,p})_-, \frac{1}{2N} \log_+ \L ] +
[(B_{\alpha,p})_+ ,\frac{1}{2M} \log_- \L ] .
\end{equation}
\end{lemma}
\begin{proof}
See \cite{C}.
 \end{proof}

From the lemma above, noticing that $[log_{+}\L,\L^{k/N}]=0$ and
$[log_{-}\L,\L^{k/M}]=0$, we can easily
get\begin{equation}\label{logL'}
 (log \L)_{t_{\alpha,p}}=[A_{\alpha,n},log \L]=
\begin{cases}
 [(B_{\alpha,n})_{+},log\L],  \ \ \ &when\  \alpha > 0,\\
 [(-B_{\alpha,n})_{-},log\L],  \ \ \ &when\  \alpha\leq 0.
  \end{cases}
  \end{equation}
 Using the lemma above, Carlet proved the following
proposition.

\begin{proposition} \label{propZS}
If $\L$ satisfies the Lax equations \eqref{edef}, then the
following Zakharov-Shabat equations hold \cite{C}
\begin{equation}
\label{zs}  (A_{\alpha,m})_{t_{\beta,n}} - (A_{\beta,
n})_{t_{\alpha,m}} + [ A_{\alpha,m} , A_{\beta,n} ] =0
\end{equation}
for $-M \leq \alpha, \beta \leq N$ ,  $m, n \geq 0$.
\end{proposition}
 Using the Zakharov-Shabat eqs.(\ref{zs}) we can prove  the following
corollary.
\begin{corollary}The following  relation holds
\begin{equation}
[\d_{t_{\beta,n}}, \d_{t_{\alpha,m}}]\L=0
\end{equation}
for $-M \leq \alpha, \beta \leq N$ ,  $m, n \geq 0$.
\end{corollary}
After the corollary above, we can prove the following lemma using the method in \cite{UT}.
\begin{lemma} \label{zs corollary}The following two equations hold
\begin{align}\label{compatible zero curvature}
 &\partial_{\beta,n}(B_{\alpha,m})_{-} - \partial_{\alpha,m}(B_{\beta,
n})_{-} - [ (B_{\alpha,m})_{-} , (B_{\beta,n})_{-} ] =
0,\\
 &-\partial_{\beta,n}(B_{\alpha,m})_{+} + \partial_{\alpha,m}(B_{\beta,
n})_{+} - [ (B_{\alpha,m})_{+} , (B_{\beta,n})_{+} ] =0
\end{align}
here, $-M\leq\alpha,\beta\leq N$,  $m,n \geq 0$.
\end{lemma}
Proof: We now only give the proof of a case of eqs.(\ref{compatible
zero curvature})  which should be taken special care of because of
the logarithm. As eqs.\eqref{zs}, $$  \partial_{-M ,n}(A_{\beta,m})
- \partial_{\beta,m}(A_{-M, n})+ [  A_{\beta,m},A_{-M,n} ] =0$$
where $-M+1\leq\beta\leq 0,$
 i.e.
$$-\partial_{-M ,n}(B_{\beta,m})_{-} - \partial_{\beta,m}(B_{-M,
n})_{+} + [ -(B_{\beta,m})_{-} , (B_{-M,n})_{+} ] =0.$$
Eqs.\eqref{dd} lead to
\begin{eqnarray}
\label{beta dd}\partial_{\beta,m}\L^n=[ -(B_{\beta,m})_{-},
\L^n]\end{eqnarray} Considering to eqs.\eqref{logL'} and using
eqs.\eqref{beta dd}, we get
\begin{eqnarray*}
\partial_{\beta,m}(B_{-M,
n})&=&\partial_{\beta,m}(\frac{2}{\epsilon
n!}[\L^n(log\L-\frac{1}{2}(\frac{1}{M}+\frac{1}{N})\C_n)])\\
&=&[-(B_{\beta,m})_{-},\frac{2}{\epsilon
n!}\L^n[log\L-\frac{1}{2}(\frac{1}{M}+\frac{1}{N})\C_n]]\\
&=&[-(B_{\beta,m})_{-},B_{-M,n}].
\end{eqnarray*}
 Then eqs.\eqref{zs} imply
\begin{eqnarray*}
0&=&[\d_{-M,n}-(B_{-M,n})_{+},\d_{\beta,m}+(B_{\beta,m})_{-}]\\
&=&[\d_{-M,n}+(B_{-M,n})_{-}-B_{-M.n},\d_{\beta,m}+(B_{\beta,m})_{-}]\\
&=&[\d_{-M,n}+(B_{-M,n})_{-},\d_{\beta,m}+(B_{\beta,m})_{-}]
+[\d_{\beta,m}+(B_{\beta,m})_{-},B_{-M,n}]\\&=&[\d_{-M,n}+(B_{-M,n})_{-},\d_{\beta,m}+(B_{\beta,m})_{-}].
\end{eqnarray*}

This is just
$$\partial_{-M ,n}(B_{\beta,m})_{-} -
\partial_{\beta,m}(B_{-M, n})_{-} + [ (B_{-M,n})_{-} , (B_{\beta,m})_{-}
] =0.$$
 One  can further verify other
identities easily by the same way.   \qed  \\
 Considering the lemma
above we can prove the following theorem.
\begin{theorem}
\label{t1} $\L$ is a solution to the EBTH if and only if there is
a pair of dressing operators $\P_L$ and $\P_R$, which satisfies
the following Sato  equations{\allowdisplaybreaks}
\begin{eqnarray}
\label{bn1}
\d_{\alpha,n}\P_L & =- ( \B_{\alpha,n}  )_- \P_L, \\
\label{bn1'} \d_{\alpha,n}\P_R & = (  \B_{\alpha ,n})_+\P_R,
\end{eqnarray}
where, $-M\leq\alpha\leq N$,  $n \geq 0$.
\end{theorem}
Proof: Using lemma \ref{zs corollary} and a
 standard procedure
given by \cite{UT} and \cite{M}, we can prove the  theorem.\\

 Sato equations can be regarded as the definitions of the wave operators, i.e. $\P_L$  and $\P_R$ in
eq.\eqref{bn1} and eq.\eqref{bn1'}.   It is unique up to multiplying
$\P_L$  and $\P_R$ from the right by operators of the form
$1+a_1\Lambda^{-1}+a_2\Lambda^{-2}+\ldots$ and $\tilde a_0 +\tilde
a_1\Lambda+\tilde a_2\Lambda^2+\ldots$ respectively, where $a_i$ and
$\tilde a_j$ are independent of $x$ and $t_{\alpha, n}$ where
$-M\leq\alpha\leq N$,  $n \geq 0$. We shall study identities related
to the wave operators in next section. On  the other hand, we shall
show relations  between  tau function and $w_i$, $\tilde{w}_i$ from
Sato eqtuaions later.
%%%%%%%%%%%%%%%%%%%%%%%%%%%%%%%
\sectionnew{Hirota Bilinear Identities of Wave Operators   }

We suppose the wave operators $\P_L$, $\P_R$ and $\P_L^{-1}$,
$\P_R^{-1}$ given by eq.(\ref{dressP}) to eq(\ref{dressQinv}),
then define the symbols $P_L$,  $P_R$ and $P_L^{-1}$, $P_R^{-1}$ as
following
\begin{eqnarray}\label{symbol PL}
\P_L(x,t,\Lambda) (\lambda^{\frac{x}{\epsilon}})&=&P_L(x,t,\lambda) \lambda^{\frac{x}{\epsilon}} ,\\
 \P_R(x,t,\Lambda) (\lambda^{\frac{x}{\epsilon}})&=&P_R(x,t,\lambda)\lambda^{\frac{x}{\epsilon}},\\
 \P_L^{-1\#}(x,t,\Lambda)( \lambda^{-\frac{x}{\epsilon}})&=&P_L^{-1}(x,t,\lambda) \lambda^{-\frac{x}{\epsilon}} ,\\
 \label{symbol PR}\P_R^{-1\#}(x,t,\Lambda)
 (\lambda^{-\frac{x}{\epsilon}})&=&P_R^{-1}(x,t,\lambda)\lambda^{-\frac{x}{\epsilon}},
\end{eqnarray}
 where
$\#$ is an {\em antiinvolution} acting on the space of Laurent
series in $\Lambda$ by $x^\#=x$ and $\Lambda^\# = \Lambda^{-1}$. The
left side of  eq.\eqref{symbol PL}-eq.\eqref{symbol PR} means the
operators $\P_L,\P_R,\P_L^{-1\#},\P_R^{-1\#}$ acting on the function
$\lambda^{\pm\frac{x}{\epsilon}}$ in the bracket.
 We should
note that $\P_L^{-1}$ and $\P_R^{-1}$ are the inverse operators of $\P_L$ and
$\P_R$ respectively, but $P_L^{-1}$ and $P_R^{-1}$ are not the
inverse symbols of $P_L$ and $P_R$ respectively.

For simplicity of Hirota bilinear identities, we will introduce two
series below
\begin{eqnarray*}
 &&  \W_L(x,t,\Lambda) = \P_L(x,t,\Lambda)\times\\
 &&
 \exp\left({
\sum_{n\geq 0} \left[\sum_ {\alpha=1}^{N}\frac{\Gamma ( 2-
\frac{\alpha-1}{N})}{\Gamma (n+2 - \frac{\alpha-1}{N} )}
\frac{\Lambda^{N({n+1-\frac{\alpha-1}{N}})}}{\epsilon}t_{\alpha, n}
\right] +\sum_{n> 0}\frac{\Lambda^{nN}}{n!}\( \epsilon\partial_{x} -
\frac{1}{2}(\frac{1}{M}+\frac{1}{N})\C_n\) \frac{t_{-M,n}}{\epsilon} }\right),\\
 &&
\W_R(x,t,\Lambda) =\P_R(x,t,\Lambda)\times \\
&&\exp\left(- \sum_{n\geq 0} \left[ \sum_
{\beta=-M+1}^{0}\frac{\Gamma ( 2+ \frac{\beta}{M})}{\Gamma (n+2 +
\frac{\beta}{M} )}
\frac{\Lambda^{-M({n+1+\frac{\beta}{M}})}}{\epsilon}t_{\beta,
n}\right]+\sum_{n> 0}\frac{\Lambda^{-nM}}{n!}\( \epsilon\partial_{x}
+\frac{1}{2}(\frac{1}{M}+\frac{1}{N})\C_n\) \frac{t_{-M,n}}{\epsilon}\right).
 \end{eqnarray*}

If the series have forms
\begin{eqnarray*} \W_L(x,t,\Lambda)=\sum_{i\in \Z} a_i(x,t,
\d_x)\Lambda^i \mbox{ and } \W_R(x,t,\Lambda)=\sum_{i\in \Z}
b_i(x,t, \d_x)\Lambda^{i}, \end{eqnarray*}

\begin{eqnarray*} \W_L^{-1}(x,t,\Lambda)=\sum_{i\in \Z}\Lambda^{i} a_i'(x,t,
\d_x) \mbox{ and } \W_R^{-1}(x,t,\Lambda)=\sum_{j\in
\Z}\Lambda^{j}b_j'(x,t, \d_x),  \end{eqnarray*} then we denote their
left symbols $W_L$,  $W_R$ and right symbols $W_L^{-1}$, $W_R^{-1}$
as following
\begin{eqnarray*}
&&  W_L(x,t,\lambda) =\sum_{i\in \Z} a_i(x,t,
\d_x)\lambda^i  =P_L(x,t,\lambda)\times\\
 &&
 \exp\left({
\sum_{n\geq 0} \left[\sum_ {\alpha=1}^{N}\frac{\Gamma ( 2-
\frac{\alpha-1}{N})}{\Gamma (n+2 - \frac{\alpha-1}{N} )}
\frac{\lambda^{N({n+1-\frac{\alpha-1}{N}})}}{\epsilon}t_{\alpha, n}
\right] +\sum_{n> 0}\frac{\lambda^{nN}}{n!}\( \epsilon\partial_{x} -
\frac{1}{2}(\frac{1}{M}+\frac{1}{N})\C_n\) \frac{t_{-M,n}}{\epsilon} }\right),\\
%%%%%%%%%%%%%%%%%%%%%%%%%%%%%%%%%%%%%%%%%%%%%%%%%%%%%%%%%%%%%%%%
&&  W^{-1}_L(x,t,\lambda)=  \sum_{i\in \Z} a_i'(x,t,
\d_x)\lambda^{i} \\
&=&
 \exp\left(-{
\sum_{n\geq 0} \left[\sum_ {\alpha=1}^{N}\frac{\Gamma ( 2-
\frac{\alpha-1}{N})}{\Gamma (n+2 - \frac{\alpha-1}{N} )}
\frac{\lambda^{N({n+1-\frac{\alpha-1}{N}})}}{\epsilon}t_{\alpha, n}
\right] -\sum_{n> 0}\frac{\lambda^{nN}}{n!}\( \epsilon\partial_{x} -
\frac{1}{2}(\frac{1}{M}+\frac{1}{N})\C_n\) \frac{t_{-M,n}}{\epsilon} }\right)\\
&& \hspace{3cm}\times P^{-1}_L(x,t,\lambda),\\
%%%%%%%%%%%%%%%%%%%%%%%%%%%%%%%%%%%%%%%%%%%%%
&&
W_R(x,t,\lambda) =\sum_{i\in \Z} b_i(x,t,
\d_x)\lambda^{i} =P_R(x,t,\lambda)\times \\
&&\exp\left(- \sum_{n\geq 0} \left[ \sum_
{\beta=-M+1}^{0}\frac{\Gamma ( 2+ \frac{\beta}{M})}{\Gamma (n+2 +
\frac{\beta}{M} )}
\frac{\lambda^{-M({n+1+\frac{\beta}{M}})}}{\epsilon}t_{\beta,
n}\right]+\sum_{n> 0}\frac{\lambda^{-nM}}{n!}\( \epsilon\partial_{x}
+ \frac{1}{2}(\frac{1}{M}+\frac{1}{N})\C_n\) \frac{t_{-M,n}}{\epsilon}\right)\\
%%%%%%%%%%%%%%%%%%%%%%%%%%%%%%%%%%%
&& W^{-1}_R(x,t,\lambda)=\sum_{j\in
\Z}b_j'(x,t, \d_x)\lambda^{j}\\
& =& \exp\left(- \sum_{n\geq 0} \left[ \sum_
{\beta=-M+1}^{0}\frac{\Gamma ( 2+ \frac{\beta}{M})}{\Gamma (n+2 +
\frac{\beta}{M} )}
\frac{\lambda^{-M({n+1+\frac{\beta}{M}})}}{\epsilon}t_{\beta,
n}\right]+\sum_{n> 0}\frac{\lambda^{-nM}}{n!}\( \epsilon\partial_{x}
+ \frac{1}{2}(\frac{1}{M}+\frac{1}{N})\C_n\) \frac{t_{-M,n}}{\epsilon}\right)\\
&& \hspace{ 3cm}\times P^{-1}_R(x,t,\lambda).
\end{eqnarray*}
 These
operator-valued symbols are quite different from common symbols
because $\epsilon\d_x$  is not replaced by its corresponding symbol $\log
\lambda$.\\
After defining residue as $\res_{\lambda }\sum_{n\in \Z}\alpha_n \lambda^n=\alpha_{-1}$, we get the following proposition using the similar proof as \cite{UT} and \cite{M}.
\begin{proposition}\label{wave-operators}
Let $t$ and $t'$ be time sequences such that $t_{-M,0}=t'_{-M,0}$.
 $\P_L$ and
$\P_R$ are wave operators of the EBTH if and only if for all   $m\in
\Z$, $r\in \N( including  \  0 )$ , the following Hirota bilinear identity hold
\begin{eqnarray}  \notag &&\res_{\lambda }
 \left\{
\lambda^{Nr+m-1}\ W_L(x,t,\epsilon \d_x,\lambda) W_L^{-1}(x-m\epsilon,t', \epsilon \d_x,\lambda)
\right\} = \\ \label{HBE3}&& \res_{\lambda }
 \left\{
\lambda^{-Mr+m-1} W_R( x,t,\epsilon \d_x,\lambda )\
W_R^{-1}(x-m\epsilon,t',\epsilon \d_x,\lambda) \right\}.
\end{eqnarray}
\end{proposition}
{\bf Proof.} \\ ($\Longrightarrow$): Set
$\ga=(\ga_{N,0},\ga_{N,1},\ga_{N,2},\ldots;\ga_{N-1,0},\ga_{N-1,1},\ga_{N-1,2},\ldots;\ldots;\ga_{-M,1},\ga_{-M,2},\ldots
)$ be a multi index and
$$
\d^\ga:=\d_{N,0}^{\ga_{N,0}}\d_{N,1}^{\ga_{N,1}}\d_{N,2}^{\ga_{N,2}}\ldots;
\d_{N-1,0}^{\ga_{N-1,0}}\d_{N-1,1}^{\ga_{N-1,1}}\d_{N-1,2}^{\ga_{N-1,2}}\ldots;\ldots;
\d_{-M,1}^{\ga_{-M,1}}\d_{-M,2}^{\ga_{-M,2}}\ldots\ ,
$$
where $\d_{\alpha,i}=\d/\d t_{\alpha,i}$ ( we stress that $\d/\d
t_{-M,0}$ is not involved). Firstly we shall prove the left
statement leads to
\begin{eqnarray} \label{HBE2} \W_L (x,t,\Lambda)\Lambda^{Nr}
\W_L^{-1}(x,t',\Lambda) = \W_R (x,t,\Lambda)
\Lambda^{-Mr}\W_R^{-1}(x,t',\Lambda)
\end{eqnarray}
 for all integers $r\geq 0$.
Just the same as the method used in\cite{M}, by induction on $\ga $,
we shall prove that
\begin{equation} \label{2.7}
\W_L(x,t,\Lambda)\Lambda^{Nr}(\d^\ga\W_L^{-1}(x,t,\Lambda))
=\W_R(x,t,\Lambda) \Lambda^{-Mr}(\d^\ga\W_R^{-1}(x,t,\Lambda)).
\end{equation} When $\ga=0$, eq.\eqref{2.7} becomes
\begin{equation} \label{2.8}
\P_L(x,t,\Lambda)\Lambda^{Nr}\P_L^{-1}(x,t,\Lambda)
=\P_R(x,t,\Lambda)\Lambda^{-Mr}\P_R^{-1}(x,t,\Lambda).
\end{equation}
which is obviously true according to the definition of
wave operators.\\
Suppose eq.\eqref{2.7} is true in the case of $\alpha \neq 0$.
 Note that

\begin{equation}
\notag\partial_{\alpha , n}\W_{L} :=
\begin{cases}
 [(\d_{\alpha,n}\P_L)\P_L^{-1}+\P_L\frac{\Gamma(2-\frac{\alpha-1}{N})}{\epsilon\Gamma(n+2-\frac{\alpha-1}{N})}
  \Lambda^{N(n+1-\frac{\alpha-1}{N})}\P_L^{-1}]\W_L,
&\alpha=N,N-1,\dots,1,\\
 (\d_{\alpha,n}\P_L)\P_L^{-1}\W_L,& \alpha = 0\dots -M+1,\\
[(\d_{\alpha,n}\P_L)\P_L^{-1}+\P_L\frac{\Lambda^{nN}}{\epsilon
n!}(\epsilon\d_x-\frac{1}{2}(\frac{1}{M}+\frac{1}{N})
\C_n)\P_L^{-1}]\W_L,  &\alpha=-M,
  \end{cases}
\end{equation}
and
\begin{align}
\notag\partial_{\alpha , n}\W_{R} :=
\begin{cases}
(\d_{\alpha,n}\P_R)\P_R^{-1}\W_R,&\alpha=N\dots 1,\\
[(\d_{\alpha,n}\P_R)\P_R^{-1}-\P_R\frac{\Gamma(2+\frac{\alpha}{M})}{\epsilon\Gamma(n+2+\frac{\alpha}{M})}
  \Lambda^{-M(n+1+\frac{\alpha}{M})}\P_R^{-1}]\W_R,
& \alpha=0,\dots,-M+1, \\
[(\d_{\alpha,n}\P_R)-\P_R\frac{\Lambda^{-nM}}{\epsilon
n!}(-\epsilon\d_x-\frac{1}{2}(\frac{1}{M}+\frac{1}{N})
\C_n)\P_R^{-1} ] \W_R,&\alpha=-M.
\end{cases}
  \end{align}
By computation we get \\

\begin{equation} \notag
 \partial_{\alpha ,
n}\W_{L} :=
\begin{cases}
(B_{\alpha,n})_{+}\W_{L},&\alpha=N\dots 1,\\
 -(B_{
\alpha,n})_{-}\W_{L}, &\alpha = 0\dots -M+1,\\
[-(B_{ \alpha,n})_{-}+\frac{1}{\epsilon n!} [ \L^n (\frac{1}{N}\log
_{+}\L - \frac12 ( \frac1M + \frac1N) \C_n ) ]]\W_{L}, &\alpha=-M,
\end{cases}
\end{equation}
\begin{equation} \notag
\partial_{\alpha,n}\W_{R}:=
\begin{cases}
(B_{\alpha,n})_{+}\W_{R}, &\alpha=N\dots 1,\\
 -(B_{\alpha,n})_{-}\W_{R}, &\alpha = 0\dots -M+1,\\
 [(B_{ \alpha,n})_{+}-\frac{1}{\epsilon n!} [ \L^n (\frac{1}{M}\log _{-}\L+
\frac12 ( \frac1M + \frac1N) \C_n ) ]]\W_{R}, &\alpha=-M,
  \end{cases}
\end{equation}
which implies
\begin{equation} \notag (\d_{\alpha,n}\W_L)\Lambda^{Nr}(\d^\ga\W_L^{-1}) =
(\d_{\alpha,n}\W_R)\Lambda^{-Mr}(\d^\ga\W_R^{-1}) \end{equation}
by considering \eqref{2.7}.
Furthermore we  get  \begin{equation} \notag
\W_L\Lambda^{Nr}(\d_{\alpha,n}\d^\ga\W_L^{-1} )=
\W_R\Lambda^{-Mr}(\d_{\alpha,n}\d^\ga\W_R^{-1}).
\end{equation} Thus if we increase the power of  $\d_{\alpha,n}$ by 1 then
eq.\eqref{2.7} still holds.
 The induction is completed.
 Using the Taylor's formula and eq.\eqref{2.7},
expanding Both sides of eq.\eqref{HBE2} about $t=t'$, we can finish
the proof of eq.\eqref{HBE2}.

Then we shall prove the right-side  statement of the proposition is equivalent to identity
eq.\eqref{HBE3}.\\
 Let $m\in \Z$, $r\in \N$ and $t_{-M,0}=t'_{-M,0}$.
Put
\begin{eqnarray*} \W_L(x,t,\Lambda)=\sum_{i\in \Z} a_i(x,t,
\d_x)\Lambda^i \mbox{ and } \W_R(x,t,\Lambda)=\sum_{i\in \Z}
b_i(x,t, \d_x)\Lambda^{i}, \end{eqnarray*}

\begin{eqnarray*} \W_L^{-1}(x,t,\Lambda)=\sum_{i\in \Z}\Lambda^{i} a_i'(x,t,
\d_x) \mbox{ and } \W_R^{-1}(x,t,\Lambda)=\sum_{j\in
\Z}\Lambda^{j}b_j'(x,t, \d_x)  \end{eqnarray*} and compare the
coefficients in front of $\Lambda^{-m}$ in eq.\eqref{HBE2}:
\begin{eqnarray*} \sum_{i+j=-m-Nr} a_i(x,t, \d_x)a_j'(x-m\epsilon,t',
\d_x) = \sum_{i+j=-m+Mr}  b_i(x,t, \d_x)b_j'(x-m\epsilon,t', \d_x).
\end{eqnarray*} This equality can be written also as \begin{eqnarray} &&
\notag \res_{\lambda }
 \left\{\lambda^{Nr+m-1}\
W_L(x,t,\epsilon \d_x,\lambda) W_L^{-1}(x-m\epsilon,t', \epsilon \d_x,\lambda) \right\} = \\
\notag && \res_{\lambda }
 \left\{\lambda^{-Mr+m-1} W_R(
x,t,\epsilon \d_x,\lambda ) W_R^{-1}(x-m\epsilon,t',\epsilon \d_x,\lambda) \right\}.
\end{eqnarray}

($\Longleftarrow$): We have proved that eq.\eqref{HBE3} is
equivalent to eq.\eqref{HBE2}. Now we will prove eq.\eqref{HBE2}
implies that operators  $\P_L$ and $\P_R$ are wave operators of the
EBTH.\\
 Differentiate eq.\eqref{HBE2} with respect to $t_{\alpha,n}$
and then put $t=t'$, we can get

\begin{eqnarray*} \notag  (\d_{\alpha,n}\P_L)\P_L^{-1}+\P_L\C_{\alpha,n}\P_L^{-1}=
 (\d_{\alpha,n}\P_R)\P_R^{-1} -\P_R\C'_{\alpha,n}\P_R^{-1}
\end{eqnarray*}
where
\begin{equation} \notag
 C_{\alpha , n} :=
\begin{cases}
  \frac{\Gamma ( 2- \frac{\alpha-1}{N} )}{\epsilon\Gamma (n+2 - \frac{\alpha-1}{N} )} \Lambda^{N(n+1-\frac{\alpha-1}{N})}, &\alpha=N\dots 1,\\
       0,                            &\alpha = 0\dots -M+1,\\
 \frac{1}{\epsilon n!} [ \Lambda^{nN}( \epsilon\d_{x} -\frac{1}{2}(\frac{1}{M}+\frac{1}{N})\C_n) ], & \alpha =
 -M,
  \end{cases}
\end{equation}
\begin{equation} \notag
 C'_{\alpha , n} :=
\begin{cases}
 0,                          &\alpha=N\dots 1,\\
  \frac{\Gamma ( 2 + \frac{\alpha}{M} )}{\epsilon\Gamma (n+2 + \frac{\alpha}{M} )} \Lambda^{-M(n+1+\frac{\alpha}{M})}, &\alpha = 0\dots -M+1,\\
 \frac{1}{\epsilon n!} [ \Lambda^{-nM}( -\epsilon\d_{x} -\frac{1}{2}(\frac{1}{M}+\frac{1}{N})\C_n) ], & \alpha = -M.
  \end{cases}
\end{equation}
 Since $(\d_{\alpha,n}\P_L)\P_L^{-1}$ contains only
negative powers of $\Lambda$ and $(\d_{\alpha,n}\P_R )\P_R^{-1}$
contains non-negative powers, we get eq.\eqref{bn1},
eq.\eqref{bn1'} by separating the negative and the positive part
of the equation. Thus
$\P_L, \ \P_R$ is a pair of wave operators. This is the end the proof.   \qed\\
 Although in the HBI eq.\eqref{HBE3} the  symbols are  not scaled-valued, we can also think about the
scalar-valued form of the  HBI.
\begin{proposition}\label{HBEscaled}
Let $1\leq \alpha \leq N, -M+1\leq\beta \leq 0 $,   $m\in \Z$, $r\in
\N$; HBI eq.\eqref{HBE3} leads to  the following scalar-valued Hirota
bilinear identities
\begin{eqnarray}  \notag &&\res_{\lambda }
 \left\{
\lambda^{Nr+m-1}[(\d_{\alpha,n}P_L(x,t,\lambda))
P_L^{-1}(x-m\epsilon,t, \lambda)+\frac{\Gamma ( 2-
\frac{\alpha-1}{N})}{\Gamma (n+2 - \frac{\alpha-1}{N} )}
\frac{\lambda^{N({n+1-\frac{\alpha-1}{N}})}}{\epsilon}P_L(x,t,\lambda)\right. \\
\label{HBEscaled3}&& \left.  P_{L}^{-1}(x-m\epsilon,t,
\lambda)]\right\} = \res_{\lambda }
 \left\{
\lambda^{-Mr+m-1} (\d_{\alpha,n}P_R( x,t,\lambda ))\
P_R^{-1}(x-m\epsilon,t,\lambda) \right\},
\end{eqnarray}
\begin{eqnarray}  \notag &&\res_{\lambda }
 \left\{
\lambda^{Nr+m-1}(\d_{\beta,n}P_L(x,t,\lambda))
P_L^{-1}(x-m\epsilon,t, \lambda)\right\} = \res_{\lambda }
\left\{ \lambda^{-Mr+m-1} \big[(\d_{\beta,n}P_R( x,t,\lambda
))\right. \\
\label{HBEscaled3'}&& \left.\
P_R^{-1}(x-m\epsilon,t,\lambda)-\frac{\Gamma ( 2+
\frac{\beta}{M})}{\Gamma (n+2 + \frac{\beta}{M} )}
\frac{\lambda^{-M({n+1+\frac{\beta}{M}})}}{\epsilon}P_R(x,t,\lambda)
P_{R}^{-1}(x-m\epsilon,t, \lambda)\big] \right\},
\end{eqnarray}
\begin{eqnarray}  \notag &&\res_{\lambda }
 \left\{
\lambda^{Nr+m-1}[(\d_{-M,n}P_L(x,t,\lambda)) P_L^{-1}(x-m\epsilon,t,
\lambda)+\frac{\lambda^{nN}}{ n!} P_L(x,t,\lambda)
P_{Lx}^{-1}(x-m\epsilon,t, \lambda) \right.
\\\notag
 && \left. -\frac{\lambda^{nN}}{\epsilon n!}\frac{1}{2}(\frac{1}{M}+\frac{1}{N})\C_n P_L(x,t,\lambda)
P_{L}^{-1}(x-m\epsilon,t, \lambda)]\right\} = \\\notag &&
\res_{\lambda }
 \left\{ \lambda^{-Mr+m-1} [(\d_{-M,n}P_R(
x,t,\lambda ))\
P_R^{-1}(x-m\epsilon,t,\lambda)+\frac{\lambda^{-nM}}{ n!}
P_R(x,t,\lambda) P_{Rx}^{-1}(x-m\epsilon,t, \lambda) \right.
\\\label{HBEscaled3''}
 && \left. +\frac{\lambda^{-nM}}{\epsilon n!}\frac{1}{2}(\frac{1}{M}+\frac{1}{N})\C_n P_R(x,t,\lambda)
P_{R}^{-1}(x-m\epsilon,t, \lambda)] \right\},
\end{eqnarray}
\begin{eqnarray}  \notag &&\res_{\lambda }
 \left\{
\lambda^{Nr+m-1}P_L(x,t,\lambda) P_L^{-1}(x-m\epsilon,t,
\lambda)\right\} \\\label{HBEscaled3'''} &=& \res_{\lambda }
\left\{ \lambda^{-Mr+m-1} P_R( x,t,\lambda )\
P_R^{-1}(x-m\epsilon,t,\lambda) \right\}.
\end{eqnarray}
\end{proposition}

\begin{proof} Let operators in both sides of eq.\eqref{HBE3} act  on ``$1$'', because
\begin{eqnarray}\notag
  \exp(\sum_{n>
0}\frac{\lambda^{nN}}{
n!}(t_{-M,n}-t'_{-M,n})\d_x)P_L^{-1}(x-m\epsilon,t',\lambda) 1&=&
P_L^{-1}(x+\sum_{n>
0}\frac{\lambda^{nN}}{
n!}(t_{-M,n}-t'_{-M,n})-m\epsilon,t',\lambda),\\
\notag   \exp(\sum_{n>
0}\frac{\lambda^{-nM}}{
n!}(t_{-M,n}-t'_{-M,n})\d_x)P_R^{-1}(x-m\epsilon,t',\lambda) 1&=&
 P_R^{-1}(x+\sum_{n>
0}\frac{\lambda^{-nM}}{
n!}(t_{-M,n}-t'_{-M,n})-m\epsilon,t',\lambda),
 \end{eqnarray}
 therefore the HBI eq.\eqref{HBE3} becomes
\begin{eqnarray}\notag
 && \res_{\lambda }
 \left\{\lambda^{Nr+m-1}P_L(x,t,\lambda)\exp(\sum_{n\geq 0} \sum_ {\alpha=1}^{N}\frac{\Gamma ( 2-
\frac{\alpha-1}{N})}{\Gamma (n+2 - \frac{\alpha-1}{N} )}
\frac{\lambda^{N({n+1-\frac{\alpha-1}{N}})}}{\epsilon}(t_{\alpha,n}-t'_{\alpha,n})-\right.\\
\notag && \left.\sum_{n>0}\frac{\lambda^{nN}} {\epsilon
n!}\frac{1}{2}(\frac{1}{M}+\frac{1}{N})\C_n(t_{-M,n}-t'_{-M,n}))P_L^{-1}(x+\sum_{n>
0}\frac{\lambda^{nN}}{
n!}(t_{-M,n}-t'_{-M,n})-m\epsilon,t',\lambda)\right\}=\\
\notag &&  \res_{\lambda }
\left\{\lambda^{-Mr+m-1}P_R(x,t,\lambda)\exp( -\sum_{n\geq 0} \sum_
{\beta=-M+1}^{0}\frac{\Gamma ( 2+ \frac{\beta}{M})}{\Gamma (n+2 +
\frac{\beta}{M} )}
\frac{\lambda^{-M({n+1+\frac{\beta}{M}})}}{\epsilon}(t_{\beta,n}-t'_{\beta,n})+\right.\\ \notag
&& \left. \sum_{n>0}\frac{\lambda^{-nM}}{\epsilon
n!}\frac{1}{2}(\frac{1}{M}+\frac{1}{N})\C_n(t_{-M,n}-t'_{-M,n}))P_R^{-1}(x+\sum_{n>
0}\frac{\lambda^{-nM}}{
n!}(t_{-M,n}-t'_{-M,n})-m\epsilon,t',\lambda)\right\}.\\\label{HBEscaledtotal}
 \end{eqnarray}
  To get eq.\eqref{HBEscaled3}, we differentiate  both sides of eq.\eqref{HBEscaledtotal} by $t_{\alpha,n}$ and
 let $t=t'$.
  To get eq.\eqref{HBEscaled3'}, we differentiate  both sides of eq.\eqref{HBEscaledtotal}  by $t_{\beta,n}$ and
 let $t=t'$.
 To get eq.\eqref{HBEscaled3''}, we differentiate   both sides of eq.\eqref{HBEscaledtotal}  by $t_{-M,n}$ and
 let $t=t'$.
   To get  eq.\eqref{HBEscaled3'''}, we just
 let $t=t'$  in eq.\eqref{HBEscaledtotal}.
\end{proof}
Moreover, HBI\eqref{HBE3} can imply other interesting identities.
\begin{proposition}\label{HBEscaledx}
Let $1\leq \alpha \leq N, -M+1\leq\beta \leq 0 $,   $r\in \N$ and
$x-x'=m\epsilon$, $m\in \Z$,  HBI \eqref{HBE3} leads to the
following scalar-valued Hirota bilinear identities
\begin{eqnarray}  \notag &&\res_{\lambda }
 \left\{
\lambda^{Nr-1}[(\d_{\alpha,n}P_L(x,t,\lambda)) P_L^{-1}(x',t,
\lambda)\lambda^{\frac{x-x'}{\epsilon}}+\frac{\Gamma ( 2- \frac{\alpha-1}{N})}{\Gamma
(n+2 - \frac{\alpha-1}{N} )}
\frac{\lambda^{N({n+1-\frac{\alpha-1}{N}})}}{\epsilon}P_L(x,t,\lambda)\right. \\
\label{HBEscaledx3}&& \left.  P_{L}^{-1}(x',t,
\lambda)\lambda^{\frac{x-x'}{\epsilon}}]\right\} = \res_{\lambda }
 \left\{
\lambda^{-Mr-1} (\d_{\alpha,n}P_R( x,t,\lambda ))\
P_R^{-1}(x',t,\lambda)\lambda^{\frac{x-x'}{\epsilon}} \right\},
\end{eqnarray}
\begin{eqnarray}  \notag &&\res_{\lambda }
 \left\{
\lambda^{Nr-1}(\d_{\beta,n}P_L(x,t,\lambda)) P_L^{-1}(x',t,
\lambda)\lambda^{\frac{x-x'}{\epsilon}}\right\} = \res_{\lambda }
 \left\{
\lambda^{-Mr-1} [(\d_{\beta,n}P_R( x,t,\lambda ))\right. \\
\label{HBEscaledx3}&& \left. P_R^{-1}(x',t,\lambda)\lambda^{\frac{x-x'}{\epsilon}}
-\frac{\Gamma ( 2+ \frac{\beta}{M})}{\Gamma (n+2 + \frac{\beta}{M}
)}
\frac{\lambda^{-M({n+1+\frac{\beta}{M}})}}{\epsilon}P_R(x,t,\lambda)
P_{R}^{-1}(x',t, \lambda)\lambda^{\frac{x-x'}{\epsilon}}] \right\},
\end{eqnarray}
\begin{eqnarray}  \notag &&\res_{\lambda }
 \left\{
\lambda^{Nr-1}[(\d_{-M,n}P_L(x,t,\lambda)) P_L^{-1}(x',t,
\lambda)\lambda^{\frac{x-x'}{\epsilon}}+\frac{\lambda^{nN}}{ n!}
P_L(x,t,\lambda) P_{Lx'}^{-1}(x',t, \lambda)\lambda^{\frac{x-x'}{\epsilon}} \right.
\\\notag && \left. -\frac{\lambda^{nN}}{\epsilon n!}\frac{1}{2}(\frac{1}{M}+\frac{1}{N})\C_n P_L(x,t,\lambda)
P_{L}^{-1}(x',t, \lambda)\lambda^{\frac{x-x'}{\epsilon}} ]\right\} = \\\notag && \res_{\lambda =
\infty} \left\{ \lambda^{-Mr-1} [(\d_{-M,n}P_R( x,t,\lambda ))\
P_R^{-1}(x',t,\lambda)\lambda^{\frac{x-x'}{\epsilon}} +\frac{\lambda^{-nM}}{ n!}
P_R(x,t,\lambda) P_{Rx'}^{-1}(x',t, \lambda)\lambda^{\frac{x-x'}{\epsilon}} \right.
\\\label{HBEscaledx31}
 && \left. +\frac{\lambda^{-nM}}{\epsilon n!}\frac{1}{2}(\frac{1}{M}+\frac{1}{N})\C_n P_R(x,t,\lambda)
P_{R}^{-1}(x',t, \lambda)\lambda^{\frac{x-x'}{\epsilon}}] \right\},
\end{eqnarray}
\begin{eqnarray}  \notag &&\res_{\lambda }
 \left\{
\lambda^{Nr-1}P_L(x,t,\lambda) P_L^{-1}(x',t,
\lambda)\lambda^{\frac{x-x'}{\epsilon}}\right\} \\\label{HBEscaledx31'} &=&
\res_{\lambda }
 \left\{ \lambda^{-Mr-1} P_R( x,t,\lambda )\
P_R^{-1}(x',t,\lambda)\lambda^{\frac{x-x'}{\epsilon}} \right\}.
\end{eqnarray}
\end{proposition}

\sectionnew{the existence of Tau-functions}

 For shortness, denote by $[\lambda^{-1}]^N$, $\[\lambda\]^{M}$ the following sequences:
\begin{equation} \notag
 \[\lambda^{-1}\]^{N}_{\alpha,n} :=
\begin{cases}
  \frac{\Gamma (n+1 - \frac{\alpha-1}{N}
)}{N\Gamma ( 2- \frac{\alpha-1}{N})}
 \epsilon
\lambda^{-N(n+1-\frac{\alpha-1}{N})}, &\alpha=N,N-1,\dots 1,\\
 0, &\alpha = 0, -1\dots -(M-1),\\
 0, & \alpha = -M.
  \end{cases}
\end{equation}

\begin{equation} \notag
\[\lambda\]^{M}_{\alpha,n} :=
\begin{cases}
0, &\alpha=N, N-1,\dots 1,\\
\frac{\Gamma (n+1 +\frac{\beta}{M} )}{M\Gamma ( 2+
\frac{\beta}{M})}
 \epsilon
\lambda^{M(n+1+\frac{\beta}{M})}, &\alpha = 0, -1, \dots -(M-1),\\
 0, & \alpha = -M.
  \end{cases}
\end{equation}
A function $\tau$  depending only on the dynamical variables $t$ and
$\epsilon$ is called the  {\em \bf tau-function of the EBTH} if it
provides symbols related to wave operators as following,
\begin{eqnarray}\label{pltau}P_L: &=&
1+\frac{w_1}{\lambda}+\frac{w_2}{\lambda^2}+\ldots : =\frac{ \tau
(t_{-M,0}+x-\frac{\epsilon}{2}, t-[\lambda^{-1}]^N;\epsilon) }
     {\tau (t_{-M,0}+x-\frac{\epsilon}{2},t;\epsilon)},\\\label{pl-1tau}
P_L^{-1}:& = &1+\frac{w_1'}{\lambda}+\frac{w_2'}{\lambda^2}+\ldots
: = \frac{\tau (t_{-M,0}+x+\frac{\epsilon}{2},
t+[\lambda^{-1}]^N;\epsilon) }
     {\tau (t_{-M,0}+x+\frac{\epsilon}{2},t;\epsilon)},\\\label{prtau}
P_R:&= &\tilde w_0 + \tilde w_1\lambda+ \tilde w_2\lambda^{2}
+\ldots : = \frac{ \tau
(t_{-M,0}+x+\frac{\epsilon}{2},t+[\lambda]^M;\epsilon)}
     {\tau(t_{-M,0}+x-\frac{\epsilon}{2},t;\epsilon)},\\\label{pr-1tau}
     P_R^{-1}:&= &\tilde w_0 '+ \tilde w_1'\lambda+ \tilde
w_2'\lambda^{2} +\ldots : = \frac{
     \tau (t_{-M,0}+x-\frac{\epsilon}{2},t-[\lambda]^M; \epsilon)}
     {\tau(t_{-M,0}+x+\frac{\epsilon}{2},t;\epsilon)}.
     \end{eqnarray}
  For a given pair of wave
operators the tau-function is unique up to  a non-vanishing function
factor which is independent of $x$, $t_{-M,0}$ and
$t_{\alpha,n}$ with all $n\geq 0$ and $-M+1\leq\alpha\leq{N}$.\\

 In
this section  we shall give a transparent and detailed proof of the
{\bf existence of tau function} for the EBTH according to the Sato
theory (\cite{Sato Fay identity},\cite{DJKM}).\\
 Let $t$ and $t'$ be two different sequences of time variables with  $t_{-M,n}=t'_{-M,n}$, $n\geq 0, r=0$, then  HBI
eq.\eqref{HBE3} becomes
\begin{eqnarray}\notag
 && \res_{\lambda }
 \left\{\lambda^{m-1}P_L(x,t,\lambda)e^{\sum_{n\geq 0} \sum_ {\alpha=1}^{N}\frac{\Gamma ( 2-
\frac{\alpha-1}{N})}{\Gamma (n+2 - \frac{\alpha-1}{N} )}
\frac{\lambda^{N({n+1-\frac{\alpha-1}{N}})}}{\epsilon}(t_{\alpha,n}-t'_{\alpha,n})}
P_L^{-1}(x-m\epsilon,t',\lambda)\right\} =\\
\notag &&  \res_{\lambda }
\left\{\lambda^{m-1}P_R(x,t,\lambda)e^{ -\sum_{n\geq 0} \sum_
{\beta=-M+1}^{0}\frac{\Gamma ( 2+ \frac{\beta}{M})}{\Gamma (n+2 +
\frac{\beta}{M} )}
\frac{\lambda^{-M({n+1+\frac{\beta}{M}})}}{\epsilon}(t_{\beta,n}-t'_{\beta,n})}
P_R^{-1}(x-m\epsilon,t',\lambda)\right\}.\\\label{chan}
 \end{eqnarray}
 By a  straightforward computation, we can infer following lemma from eq.\eqref{chan},
 which are necessary for our main theorem on tau function.

 \begin{lemma}\label{log lemma}
 The following three identities hold
 \begin{eqnarray} \notag
&&\log{P_L(x,t,\lambda_1)}-\log{P_L(x,t-[\lambda_2^{-1}]^{N},\lambda_1)}\\
\label{pl=pl''} &&
=\log{P_L(x,t,\lambda_2)}-\log{P_L(x,t-[\lambda_1^{-1}]^{N},\lambda_2)}.
\end{eqnarray}
\begin{eqnarray}\notag
&&\log{P_L(x,t,\lambda_{1})}-\log{P_L(x+\epsilon,t+[\lambda_{2}]^{M},\lambda_{1})}
=\log{P_R(x,t,\lambda_{2})}-\log{P_R(x,t-[\lambda_1^{-1}]^{N},\lambda_{2})}.\\
\label{lr''1}
\end{eqnarray}
\begin{eqnarray}&& \notag
\log{P_R(x,t,\lambda_1)}-\log{P_R(x+\epsilon,t+[\lambda_2]^M,\lambda_1)} \\
 \label{prpr''} &&
=\log{P_R(x,t,\lambda_2)}-\log{P_R(x+\epsilon,t+[\lambda_1]^M,\lambda_2)}.
 \end{eqnarray}
 \end{lemma}
 \begin{proof}
For the proof of identity\eqref{pl=pl''}, we shall set $m=1,
t'=t-[\lambda_1^{-1}]^{N}-[\lambda_2^{-1}]^{N}$ in eq.\eqref{chan}.
 Using the identity
\begin{eqnarray*} \exp \Big(\sum_{n\geq 0}\sum_ {\alpha=0}^{N-1}
\frac{(\lambda_1^{-1}\lambda)^{N({n+1-\frac{\alpha}{N}})}}{N(n+1-\frac{\alpha}{N})}
\Big) =
 (1-\lambda_1^{-1}\lambda)^{-1},
\end{eqnarray*}   the
bilinear identity eq.\eqref{chan} gives
  \begin{eqnarray}\notag
 && \res_{\lambda }
 \left\{P_L(x,t,\lambda)
P_L^{-1}(x-\epsilon,t-[\lambda_1^{-1}]^{N}-[\lambda_2^{-1}]^{N},\lambda)\frac{1}{(1-\frac{\lambda}{\lambda_1})(1-\frac{\lambda}{\lambda_2})}\right\}
 =\\
\label{hunhe1}&&  \res_{\lambda }
 \left\{P_R(x,t,\lambda)
P_R^{-1}(x-\epsilon,t-[\lambda_1^{-1}]^{N}-[\lambda_2^{-1}]^{N},\lambda)\right\}.
 \end{eqnarray}
Using \begin{eqnarray*}
(1-\lambda_1^{-1}\lambda)^{-1}(1-\lambda_2^{-1}\lambda)^{-1} =
\frac{\lambda_1\lambda_2}{\lambda_2-\lambda_1}\{(1-\lambda_1^{-1}\lambda)^{-1}-
(1-\lambda_2^{-1}\lambda)^{-1}\}\lambda^{-1},
 \end{eqnarray*}
\begin{eqnarray*}
 \res_{\lambda }
 \left\{f(\lambda)\frac{1}{\lambda(1-\frac{\lambda}{\lambda_1})}\right\} =f(\lambda_1),
 \end{eqnarray*}
 where $f(\lambda)=1+\sum_{i=1}^{\infty}a_i\lambda^{-i}$ is a formal series of $\lambda$ , then
eq.\eqref{hunhe1} infers
\begin{eqnarray}\notag
&&P_L(x,t,\lambda_1)P_L^{-1}(x-\epsilon,t-[\lambda_1^{-1}]^{N}-[\lambda_2^{-1}]^{N},\lambda_1)\\
 \label{pl=pl} &&
=P_L(x,t,\lambda_2)P_L^{-1}(x-\epsilon,t-[\lambda_1^{-1}]^{N}-[\lambda_2^{-1}]^{N},\lambda_2).
\end{eqnarray}

Setting $\lambda_1=\lambda$ and $\lambda_2=\infty$, we obtain

\begin{eqnarray}\label{}
P_L(x,t,\lambda)P_L^{-1}(x-\epsilon,t-[\lambda^{-1}]^{N},\lambda)=1,
\end{eqnarray}
which is equivalent to
\begin{eqnarray}\label{plPL=1}
P_L^{-1}(x-\epsilon,t-[\lambda^{-1}]^{N},\lambda)=\frac{1}{P_L(x,t,\lambda)},
\end{eqnarray}
Using this identity, eq.\eqref{pl=pl} gives
\begin{eqnarray} \label{pl=pl'}
\frac{P_L(x,t,\lambda_1)}{P_L(x,t-[\lambda_2^{-1}]^{N},\lambda_1)}
=\frac{P_L(x,t,\lambda_2)}{P_L(x,t-[\lambda_1^{-1}]^{N},\lambda_2)}.
\end{eqnarray}
or equivalently to
eq.\eqref{pl=pl''}.

To prove identity\eqref{lr''1}, we shall set  $m=0,
t'=t-[\lambda_1^{-1}]^{N}+[\lambda_2]^{M}$ in eq.\eqref{chan}.\\
 In this case, using the identities
\begin{eqnarray*} \exp \Big(\sum_{n\geq 0}\sum_ {\alpha=0}^{N-1}
\frac{(\lambda_1^{-1}\lambda)^{N({n+1-\frac{\alpha}{N}})}}{N(n+1-\frac{\alpha}{N})}
\Big) =
 (1-\lambda_1^{-1}\lambda)^{-1},
\end{eqnarray*}

\begin{eqnarray*} \exp \left(\sum_{n\geq 0}\sum_ {\alpha=-M+1}^{0}
\frac{(\lambda_2\lambda^{-1})^{M({n+1+\frac{\alpha}{M}})}}{M(n+1+\frac{\alpha}{M})}
\right) =
 (1-\lambda_2\lambda^{-1})^{-1},
\end{eqnarray*}
the bilinear identity \eqref{chan} gives
  \begin{eqnarray}\notag
 && \res_{\lambda }
 \left\{P_L(x,t,\lambda)
P_L^{-1}(x,t-[\lambda_1^{-1}]^{N}+[\lambda_2]^{M},\lambda)\lambda^{-1}\frac{1}{1-\frac{\lambda}{\lambda_1}}\right\}
 =\\
\label{hunhe}&&  \res_{\lambda }
 \left\{P_R(x,t,\lambda)
P_R^{-1}(x,t-[\lambda_1^{-1}]^{N}+[\lambda_2]^{M},\lambda)\lambda^{-1}\frac{1}{1-\frac{\lambda_2}{\lambda}}\right\}.
 \end{eqnarray}
Consider another residue formula
\begin{eqnarray}\label{respositiveseries}
 \res_{\lambda }
 \left\{f(\lambda)\frac{1}{\lambda-\lambda_1}\right\} =f(\lambda_1),
 \end{eqnarray}
 where $f(\lambda)=a_0+\sum\limits_{i=1}^{\infty} a_i\lambda^i$ is a formal series of $\lambda$,
  eq. \eqref{hunhe} further leads to
\begin{eqnarray} \notag
&&P_L(x,t,\lambda_{1})P_L^{-1}(x,t-[\lambda_1^{-1}]^{N}+[\lambda_{2}]^{M},\lambda_{1})\\
\label{lr'}
&&=P_R(x,t,\lambda_{2})P_R^{-1}(x,t-[\lambda_1^{-1}]^{N}+[\lambda_{2}]^{M},\lambda_{2}).
\end{eqnarray}

Setting $\lambda_1=\infty$ and $\lambda_2=\lambda$ in above
equation, then
\begin{eqnarray}\label{prpr=1}
P_R(x,t,\lambda)P_R^{-1}(x,t+[\lambda]^{M},\lambda)=1,
\end{eqnarray}
\begin{eqnarray}\label{prpr=1'}
P_R^{-1}(x,t+[\lambda]^{M},\lambda)=\frac{1}{P_R(x,t,\lambda)}.
\end{eqnarray}
Using identity eq.\eqref{plPL=1}, eq.\eqref{lr'} and
eq.\eqref{prpr=1'}, we  get
\begin{eqnarray} \label{lr''}
&&\frac{P_L(x,t,\lambda_{1})}{P_L(x+\epsilon,t+[\lambda_{2}]^{M},\lambda_{1})}
=\frac{P_R(x,t,\lambda_{2})}{P_R(x,t-[\lambda_1^{-1}]^{N},\lambda_{2})},
\end{eqnarray}
which is equivalent to
eq.\eqref{lr''1}.

For proving identity\eqref{prpr''}, we  set $m=-1,t'=t+[\lambda_1]^M + [\lambda_2]^M$ in eq.\eqref{chan}.
The bilinear identity eq.\eqref{chan} gives
\begin{eqnarray}\notag
 && \res_{\lambda }
 \left\{P_L(x,t,\lambda)
P_L^{-1}(x+\epsilon,t+[\lambda_1]^{M}+[\lambda_2]^{M},\lambda)\lambda^{-2}\right\}
 =\\
\label{PRPR}&&  \res_{\lambda }
 \left\{P_R(x,t,\lambda)
P_R^{-1}(x+\epsilon,t+[\lambda_1]^{M}+[\lambda_2]^{M},\lambda)
\frac{\lambda^{-2}}{(1-\frac{\lambda_1}{\lambda})(1-\frac{\lambda_2}{\lambda})}\right\}.
 \end{eqnarray} Using formula
 \begin{eqnarray*}
(1-\lambda_1\lambda^{-1})^{-1}(1-\lambda_2\lambda^{-1})^{-1} =
\frac{1}{\lambda_1-\lambda_2}\{(1-\lambda_1\lambda^{-1})^{-1}-
(1-\lambda_2\lambda^{-1})^{-1}\}\lambda
\end{eqnarray*}
and residue formula eq.\eqref{respositiveseries}, eq.\eqref{PRPR} further gives
\begin{eqnarray}&& \notag
P_R(x,t,\lambda_1)P_R^{-1}(x+\epsilon,t+[\lambda_1]^M+[\lambda_2]^M,\lambda_1) \\
 \label{prpr} &&
=P_R(x,t,\lambda_2)P_R^{-1}(x+\epsilon,t+[\lambda_1]^M+[\lambda_2]^M,\lambda_2).
 \end{eqnarray}
Using identity \eqref{prpr=1}, eq.\eqref{prpr} leads
\begin{eqnarray}\label{prpr'} &&
\frac{P_R(x,t,\lambda_1)}{P_R(x+\epsilon,t+[\lambda_2]^M,\lambda_1)}
=\frac{P_R(x,t,\lambda_2)}{P_R(x+\epsilon,t+[\lambda_1]^M,\lambda_2)},
 \end{eqnarray}
which is equivalent to
eq.\eqref{prpr''}.
 So  the proof of the lemma is completed now.
 \end{proof}
By lemma \ref{log lemma}, we  get the following theorem.
\begin{theorem}\label{tau-function}
Given a pair of wave operators $\P_L$ and $\P_R$ of the EBTH there
exists a unique corresponding tau-function  up to a non-vanishing
function factor which is independent of $t_{-M,0}$ and
$t_{\alpha,n},\ n\geq 0,-M+1\leqslant\alpha\leqslant N-1$.
\end{theorem}
\begin{proof}
As \cite{DJKM}, the proof is a little complicated and  the  process can be divided
into there steps. For the first step, we shall  define a 1-form
$\omega$, and then give the translational invariance of $d\omega$.
Then we will prove the 1-form  is closed in the second step which
leads to the existence of tau function $\tau(t)$. The third step is
devoted  to give the certain value of integration constants such
that we can get the symbols of dressing operators by $\tau(t)$. To
this end, define
\begin{eqnarray}\notag
\omega_L(\epsilon,x,t)&: =&-\sum_ {\alpha=1}^{N}\sum_ {n\geq 0}d
 t_{\alpha,n}\res_{\lambda}\left\{\frac{\Gamma ( 2-
\frac{\alpha-1}{N})}{\Gamma (n+2 - \frac{\alpha-1}{N} )}
\frac{\lambda^{N({n+1-\frac{\alpha-1}{N}})}}{\epsilon}
(\frac{\d}{\d \lambda}+\sum_{n'\geq 0}\sum_
{\alpha'=1}^{N}\right.\\
 && \left.\frac{\Gamma (n'+2 - \frac{\alpha'-1}{N} )}{\Gamma (
2- \frac{\alpha'-1}{N})}
 \epsilon
{\lambda^{-N(n'+1-\frac{{\alpha'-2}}{N})}
 \frac{\partial}{\partial{
t_{\alpha' ,n'}}}})\log{P_L(x,t,\lambda)}\right\}.
 \end{eqnarray}

\begin{eqnarray}\notag
\omega_R(\epsilon,x,t)&:=&\sum_ {\beta=-M+1}^{0}\sum_ {n\geq 0}d
 t_{\beta,n}\res_{\lambda}\left\{\frac{\Gamma ( 2+
\frac{\beta}{M})}{\Gamma (n+2 + \frac{\beta}{M} )}
\frac{\lambda^{-M({n+1+\frac{\beta}{M}})}}{\epsilon}(\frac{\d}{\d
\lambda}-\sum_{n'\geq 0}\sum_ {\beta'=-M+1}^{0}\right.\\
 && \left.\frac{\Gamma (n'+2+
\frac{\beta'}{M} )}{\Gamma ( 2+ \frac{\beta'}{M})}
 \epsilon
{\lambda^{M(n'+1+\frac{{\beta'-1}}{M})}
 \frac{\partial}{\partial{
t_{\beta' ,n'}}}})\log{P_R(x,t,\lambda)}\right\}.
 \end{eqnarray}
Using the three identity eq.\eqref{pl=pl''}, eq.\eqref{lr''1} and
eq.\eqref{prpr''} in lemma \ref{log lemma}, we  get
\begin{eqnarray}
\label{eq1}
\omega_L(x,t)-\omega_L(x,t-[\lambda^{-1}]^N)&=&-d_L\log P_L(x,t,\lambda),\\
\label{eq2}\omega_R(x,t)-\omega_R(x,t-[\lambda^{-1}]^N)&=&-d_R\log P_L(x,t,\lambda),\\
\label{eq3}\omega_L(x,t)-\omega_L(x+\epsilon,t+[\lambda]^M)&=&-d_L\log P_R(x,t,\lambda),\\
\label{eq4}\omega_R(x,t)-\omega_R(x+\epsilon,t+[\lambda]^M)&=&-d_R\log
P_R(x,t,\lambda),
 \end{eqnarray}
where
\begin{eqnarray}\notag
d_L&=&\sum_ {\alpha=1}^{N}\sum_ {n\geq 0}d
 t_{\alpha,n}\frac{\d}{\d t_{\alpha,n}},\\
d_R&=&\sum_ {\beta=-M+1}^{0}\sum_ {n\geq 0}d
 t_{\beta,n}\frac{\d}{\d t_{\beta,n}}.
 \end{eqnarray}
 Here we only give the proof of eq.\eqref{eq1} using identity eq.\eqref{pl=pl''} in the following, the other there equations can be got in the same way.
 \begin{eqnarray*}
&&\omega_L(x,t)-\omega_L(x,t-[\lambda^{-1}]^N)\\
& =&-\sum_ {\alpha=1}^{N}\sum_ {n\geq 0}d
 t_{\alpha,n}\res_{\lambda_1}\left\{\frac{\Gamma ( 2-
\frac{\alpha-1}{N})}{\Gamma (n+2 - \frac{\alpha-1}{N} )}
\frac{\lambda_1^{N({n+1-\frac{\alpha-1}{N}})}}{\epsilon}
(\frac{\d}{\d \lambda_1}+\sum_{n'\geq 0}\sum_
{\alpha'=1}^{N}\right.\\
 && \left.\frac{\Gamma (n'+2 - \frac{\alpha'-1}{N} )}{\Gamma (
2- \frac{\alpha'-1}{N})}
 \epsilon
{\lambda_1^{-N(n'+1-\frac{{\alpha'-2}}{N})}
 \frac{\partial}{\partial{
t_{\alpha' ,n'}}}})(\log{P_L(x,t,\lambda_1)}-\log{P_L(x,t-[\lambda^{-1}]^N,\lambda_1)})\right\}\\
& =&-\sum_ {\alpha=1}^{N}\sum_ {n\geq 0}d
 t_{\alpha,n}\res_{\lambda_1}\left\{\frac{\Gamma ( 2-
\frac{\alpha-1}{N})}{\Gamma (n+2 - \frac{\alpha-1}{N} )}
\frac{\lambda_1^{N({n+1-\frac{\alpha-1}{N}})}}{\epsilon}
(\frac{\d}{\d \lambda_1}+\sum_{n'\geq 0}\sum_
{\alpha'=1}^{N}\right.\\
 && \left.\frac{\Gamma (n'+2 - \frac{\alpha'-1}{N} )}{\Gamma (
2- \frac{\alpha'-1}{N})}
 \epsilon
{\lambda_1^{-N(n'+1-\frac{{\alpha'-2}}{N})}
 \frac{\partial}{\partial{
t_{\alpha' ,n'}}}})(\log{P_L(x,t,\lambda)}-\log{P_L(x,t-[\lambda_1^{-1}]^N,\lambda)})\right\}\\
& =&-\sum_ {\alpha=1}^{N}\sum_ {n\geq 0}d
 t_{\alpha,n}\res_{\lambda_1}\left\{\frac{\Gamma ( 2-
\frac{\alpha-1}{N})}{\Gamma (n+2 - \frac{\alpha-1}{N} )}
\frac{\lambda_1^{N({n+1-\frac{\alpha-1}{N}})}}{\epsilon}
\sum_{n'\geq 0}\sum_
{\alpha'=1}^{N}\right.\\
 && \left.\frac{\Gamma (n'+2 - \frac{\alpha'-1}{N} )}{\Gamma (
2- \frac{\alpha'-1}{N})}
 \epsilon
{\lambda_1^{-N(n'+1-\frac{{\alpha'-2}}{N})}
 \frac{\partial}{\partial{
t_{\alpha' ,n'}}}}\log{P_L(x,t,\lambda)}\right\}\\
& =&-\sum_ {\alpha=1}^{N}\sum_ {n\geq 0}d
 t_{\alpha,n}
 \frac{\partial}{\partial{
t_{\alpha ,n}}}\log{P_L(x,t,\lambda)}\\
&=&-d_L\log P_L(x,t,\lambda).
 \end{eqnarray*}
 In the calculation above, we should note the following identity holds
  \begin{eqnarray*}
(\frac{\d}{\d \lambda_1}+\sum_{n'\geq 0}\sum_
{\alpha'=1}^{N}\frac{\Gamma (n'+2 - \frac{\alpha'-1}{N} )}{\Gamma (
2- \frac{\alpha'-1}{N})}
 \epsilon
{\lambda_1^{-N(n'+1-\frac{{\alpha'-2}}{N})}
 \frac{\partial}{\partial{
t_{\alpha' ,n'}}}})\log{P_L(x,t-[\lambda_1^{-1}]^N,\lambda)}=0.
 \end{eqnarray*}
 In fact equations  \eqref{eq1}-\eqref{eq4} can be seen as a generalization of  eqs.(3.16) in \cite{Fay-like}.
Moreover we define
\begin{eqnarray}\notag
d &=&d_L + d_R,\ \ \ \ \ \ \omega(x,t)=\omega_L(x,t)+\omega_R(x,t).
 \end{eqnarray}
 Eq.\eqref{eq1} and  Eq.\eqref{eq2} lead to
\begin{eqnarray}
\label{eq12} \omega(x,t)-\omega(x,t-[\lambda^{-1}]^N)&=&-d\log
P_L(x,t,\lambda).
 \end{eqnarray}
   Eq.\eqref{eq3} and  Eq.\eqref{eq4} lead to
\begin{eqnarray}
\label{eq34} \omega(x,t)-\omega(x+\epsilon,t+[\lambda]^M)&=&-d\log
P_R(x,t,\lambda).
 \end{eqnarray}
When $\lambda=0$, eq.\eqref{eq34}lead to
\begin{eqnarray}
\label{eq34'} \omega(x,t)-\omega(x+\epsilon,t)&=&-d\log \tilde
w_0(x,t).
 \end{eqnarray}
 Differentiate both sides of equations in eq.\eqref{eq12}, eq.\eqref{eq34}
 and
 eq.\eqref{eq34'}, we get

\begin{eqnarray}
\label{eq12'} d \omega(x,t)&=&d \omega(x,t-[\lambda^{-1}]^N),\\
\label{equ34'} d\omega(x,t)&=&d\omega(x+\epsilon,t+[\lambda]^M),\\
 \label{eq34''} d \omega(x,t)&=&d
\omega(x+\epsilon,t),
 \end{eqnarray}
 which shows $d \omega(x,t)$ is independent of $x,t_{\alpha,n}, -M+1\leq\alpha \leq N, n\geq 0$.
Without loss of generality, we can assume
\begin{eqnarray}
\label{eq} d \omega(x,t)&=&\sum_{\alpha,\beta=-M+1}^N\sum_{n,m\geq
0}a(\epsilon)_{\alpha,n,\beta,m}dt_{\alpha,n}\wedge dt_{\beta,m}
 \end{eqnarray}
 where $a(\epsilon)_{\alpha,n,\beta,m}$ are independent of $x,t_{\alpha,n}, -M+1\leq\alpha \leq N, n\geq 0$
 and
 $a(\epsilon)_{\alpha,n,\beta,m}=-a(\epsilon)_{\beta,m,\alpha,n}$.
 So
\begin{eqnarray}
\label{eq'} \omega(x,t)&=&\sum_{\beta=-M+1}^N\sum_{m\geq
0}(\sum_{\alpha=-M+1}^N\sum_{n\geq
0}a(\epsilon)_{\alpha,n,\beta,m}t_{\alpha,n}) dt_{\beta,m}+d
F(\epsilon,x,t)
 \end{eqnarray}
 for arbitrary function $F(\epsilon,x,t)$.
 Taking $ \omega(x,t)$ in  eq.\eqref{eq'} back into the equation \eqref{eq12} and
 \eqref{eq34}, then
\begin{eqnarray}
 \notag -d\log P_L(x,t,\lambda)&=&d F(x,t)-d
F(x,t-[\lambda^{-1}]^N)+\sum_{\beta=-M+1}^N\sum_{m\geq
0}(\sum_{\alpha=1}^N\sum_{n\geq 0}\\\label{eq1'2'}
&&a_{\alpha,n,\beta,m}\frac{\Gamma (n+1 - \frac{\alpha-1}{N}
)}{N\Gamma ( 2- \frac{\alpha-1}{N})}
 \epsilon
{\lambda^{-N(n+1-\frac{{\alpha-1}}{N})}
 }) dt_{\beta,m},
 \end{eqnarray}

\begin{eqnarray}
\notag -d\log P_R(x,t,\lambda)&=&d F(x,t)-d
F(x+\epsilon,t+[\lambda]^M)-\sum_{\beta=-M+1}^N\sum_{m\geq
0}(\sum_{\alpha=-M+1}^0\sum_{n\geq 0}\\\label{eq3'4'}
&&a_{\alpha,n,\beta,m}\frac{\Gamma (n+1 +\frac{\alpha}{M}
)}{M\Gamma ( 2+\frac{\alpha}{M})}
 \epsilon {\lambda^{M(n+1+\frac{{\alpha}}{M})}
 }) dt_{\beta,m}.
 \end{eqnarray}
Furthermore, two identities above lead to
\begin{eqnarray}
\label{eq12''} \notag&&\log P_L(x,t,\lambda)\\\notag
&=&F(x,t-[\lambda^{-1}]^N)-F(x,t)-
\sum_{\beta=-M+1}^N\sum_{m\geq 0}(\sum_{\alpha=1}^N\sum_{n\geq 0}\\
&&a_{\alpha,n,\beta,m}\frac{\Gamma (n+1 - \frac{\alpha-1}{N}
)}{N\Gamma ( 2- \frac{\alpha-1}{N})}
 \epsilon
{\lambda^{-N(n+1-\frac{{\alpha-1}}{N})}
 }) t_{\beta,m}+H_L(\epsilon, x,t_{-M,n},\lambda),
 \end{eqnarray}
\begin{eqnarray}
\label{eq34''}\notag &&\log P_R(x,t,\lambda)\\\notag &=&
F(x+\epsilon,t+[\lambda]^M)-F(x,t)+\sum_{\beta=-M+1}^N\sum_{m\geq
0}(\sum_{\alpha=-M+1}^0\sum_{n\geq
0}\\
&&a_{\alpha,n,\beta,m}\frac{\Gamma (n+1 +\frac{\alpha}{M} )}{M\Gamma
( 2+\frac{\alpha}{M})}
 \epsilon {\lambda^{M(n+1+\frac{{\alpha}}{M})}
 }) t_{\beta,m}+H_R(\epsilon, x,t_{-M,n},\lambda),
 \end{eqnarray}
where the functions $H_L(\epsilon,
x,t_{-M,n},\lambda)=\sum_{i=1}^{\infty}H_{li}(\epsilon,
x,t_{-M,n})\lambda^{-i}$ and \\
$H_R(\epsilon,
x,t_{-M,n},\lambda)=\sum_{i=0}^{\infty}H_{ri}(\epsilon,
x,t_{-M,n})\lambda^{i}$ are independent on $t_{\alpha,n} (-M+1\leq
\alpha \leq N)$. Taking these results back into eq.\eqref{pl=pl''},
eq.\eqref{lr''1} and eq.\eqref{prpr''}, then
\begin{eqnarray}
 \notag&&\sum_{\beta=1}^N\sum_{m\geq
0}(\sum_{\alpha=1}^N\sum_{n\geq 0}a_{\alpha,n,\beta,m}\frac{\Gamma
(n+1 - \frac{\alpha-1}{N} )}{N\Gamma ( 2- \frac{\alpha-1}{N})}
 \epsilon
{\lambda_1^{-N(n+1-\frac{{\alpha-1}}{N})}
 })\frac{\Gamma (m+1 - \frac{\beta-1}{N}
)}{N\Gamma ( 2- \frac{\beta-1}{N})}
 \epsilon{\lambda_2^{-N(m+1-\frac{{\beta-1}}{N})}}=\\\notag
&&\sum_{\beta=1}^N\sum_{m\geq 0}(\sum_{\alpha=1}^N\sum_{n\geq
0}a_{\alpha,n,\beta,m}\frac{\Gamma (n+1 - \frac{\alpha-1}{N}
)}{N\Gamma ( 2- \frac{\alpha-1}{N})}
 \epsilon
{\lambda_2^{-N(n+1-\frac{{\alpha-1}}{N})}
 })\frac{\Gamma (m+1 - \frac{\beta-1}{N}
)}{N\Gamma ( 2- \frac{\beta-1}{N})}
 \epsilon{\lambda_1^{-N(m+1-\frac{{\beta-1}}{N})}},\\\label{eq12'''}
 \end{eqnarray}

\begin{eqnarray}
 \notag&&\sum_{\beta=-M+1}^0\sum_{m\geq
0}(\sum_{\alpha=1}^N\sum_{n\geq 0}a_{\alpha,n,\beta,m}\frac{\Gamma
(n+1 - \frac{\alpha-1}{N} )}{N\Gamma ( 2- \frac{\alpha-1}{N})}
 \epsilon
{\lambda_1^{-N(n+1-\frac{{\alpha-1}}{N})}
 })\frac{\Gamma (m+1 + \frac{\beta}{M}
)}{M\Gamma ( 2+ \frac{\beta}{M})}
 \epsilon{\lambda_2^{M(m+1+\frac{{\beta}}{M})}}\\
 \notag&&+H_L(\epsilon, x,\lambda_1)-H_L(\epsilon,
 x+\epsilon,\lambda_1)=\\\notag
&&\sum_{\beta=1}^N\sum_{m\geq 0}(\sum_{\alpha=-M+1}^0\sum_{n\geq
0}a_{\alpha,n,\beta,m}\frac{\Gamma (n+1+\frac{\alpha}{M}
)}{M\Gamma ( 2+ \frac{\alpha}{M})}
 \epsilon
{\lambda_2^{M(n+1+\frac{{\alpha}}{M})}
 })\frac{\Gamma (m+1 - \frac{\beta-1}{N}
)}{N\Gamma ( 2- \frac{\beta-1}{N})}
 \epsilon{\lambda_1^{-N(m+1-\frac{{\beta-1}}{N})}},\\\label{eqhunhe''}
 \end{eqnarray}
 %%%%%%%%%%%%%%%%%%%%%%%%%%%%%%%%%%%%%%%%%%%%%%%%%%%%
\begin{eqnarray}
 \notag&&\sum_{\beta=-M+1}^0\sum_{m\geq
0}(\sum_{\alpha=-M+1}^0\sum_{n\geq
0}a_{\alpha,n,\beta,m}\frac{\Gamma (n+1+\frac{\alpha}{M}
)}{M\Gamma ( 2+ \frac{\alpha}{M})}
 \epsilon
{\lambda_1^{M(n+1+\frac{{\alpha}}{M})}
 })\frac{\Gamma (m+1 + \frac{\beta}{M}
)}{M\Gamma ( 2+ \frac{\beta}{M})}
 \epsilon{\lambda_2^{M(m+1+\frac{{\beta}}{M})}}=\\\notag
&&\sum_{\beta=-M+1}^0\sum_{m\geq
0}(\sum_{\alpha=-M+1}^0\sum_{n\geq
0}a_{\alpha,n,\beta,m}\frac{\Gamma (n+1+\frac{\alpha}{M}
)}{M\Gamma ( 2+ \frac{\alpha}{M})}
 \epsilon
{\lambda_2^{M(n+1+\frac{{\alpha}}{M})}
 })\frac{\Gamma (m+1 + \frac{\beta}{M}
)}{M\Gamma ( 2+ \frac{\beta}{M})}
 \epsilon{\lambda_1^{M(m+1+\frac{{\beta}}{M})}}.\\\label{eq34i'''}
 \end{eqnarray}
 By comparing the
coefficients on both sides of
eq.\eqref{eq12'''}-eq.\eqref{eq34i'''}, we get
$a_{\alpha,n,\beta,m}=a_{\beta,m,\alpha,n}$. With the help of the
defining condition in eq.\eqref{eq}, i.e,
$a_{\alpha,n,\beta,m}=-a_{\beta,m,\alpha,n}$,
$a_{\alpha,n,\beta,m}=0$ hold for all $-M+1\leq \alpha,\beta\leq N;
n,m\geq 0$. So from identity \eqref{eq}, we have $d \omega(x,t)=0$.
 We thus  conclude that there {\bf exists a
non-vanishing function $\tau(\epsilon,x,t)$} such that
\begin{eqnarray}
\label{eq34'''} \omega(\epsilon, x,t)&=&d \log
\tau(x-\frac{\epsilon}{2},t).
 \end{eqnarray}
 In fact the function $\tau(x-\frac{\epsilon}{2}, t)$ can be written in another form as  $\tau(t_{-M,0}+x-\frac{\epsilon}{2},\bar t)$, where $\bar t$ is denoted as all the other time variables except $t_{-M,0}$.
  Therefore eq.\eqref{eq34'} can be rewritten as
 \begin{eqnarray}\label{w0tau} \tilde w_0(\epsilon,
x,t)&=&\frac{\tau(x+\frac{\epsilon}{2},t)}{\tau(x-\frac{\epsilon}{2},t)}.
 \end{eqnarray}

From eq.\eqref{eq'}, we can take $F(\epsilon, x,t)=\log
\tau(x-\frac{\epsilon}{2},t)$. So eq.\eqref{eq12''} and
eq.\eqref{eq34''} give us
\begin{eqnarray}\notag
\log P_L(x,t,\lambda)
&=&\log\tau(x-\frac{\epsilon}{2},t-[\lambda^{-1}]^N)-\log\tau(x-\frac{\epsilon}{2},t)+
H_L(\epsilon, x,t_{-M,n},\lambda),\\ \label{eq12''''}   \\ \notag
\log P_R(x,t,\lambda) &=&
\log\tau(x+\frac{\epsilon}{2},t+[\lambda]^M)-\log\tau(x-\frac{\epsilon}{2},t)+H_R(\epsilon,
x,t_{-M,n},\lambda).\\\label{eq34''''}
 \end{eqnarray}
 Substituting these into the definition of $\omega$ and using
 eq.\eqref{eq34'''} we will see that $H_L(\epsilon, x,\lambda), H_R(\epsilon,
 x,\lambda)$ are all zero. eq.\eqref{eq12''''}, eq.\eqref{plPL=1}, eq.\eqref{eq34''''} and  eq.\eqref{prpr=1'}
 will give  the equations eq.\eqref{pltau},  eq.\eqref{pl-1tau}, eq.\eqref{prtau} and eq.\eqref{pr-1tau} in the definition of $\tau$.
So the proof of existence of tau function is finished.
\end{proof}

Next we shall consider the Fay-like idenities on the tau functions.
To this end, by taking the definition  of tau function in
\eqref{pltau},  \eqref{pl-1tau}, \eqref{prtau} and \eqref{pr-1tau}
into eq.\eqref{chan} and  replacing $x-\frac{\epsilon}{2}$ by $x$ in the tau
function, we  get the following
Hirota bilinear identity
\begin{eqnarray}\notag
&& \res_{\lambda }
 \left\{ \lambda^{m-1}
\tau(x,t-[\lambda^{-1}]^N)\times
\tau(x-(m-1)\epsilon,t'+[\lambda^{-1}]^N) e^{\xi_L(t-t')} \right\}
\\\label{HBE511} &&=  \res_{\lambda }
 \left\{ \lambda^{m-1}
\tau(x+\epsilon,t+[\lambda]^M)\times
\tau(x-m\epsilon,t'-[\lambda]^M) e^{\xi_R(t-t')}\right\} ,
\end{eqnarray}
where\begin{eqnarray*} \xi_L(t-t')&=&\sum_{n\geq 0} \sum_
{\alpha=1}^{N}\frac{\Gamma ( 2- \frac{\alpha-1}{N})}{\Gamma (n+2 -
\frac{\alpha-1}{N} )}
\frac{\lambda^{N({n+1-\frac{\alpha-1}{N}})}}{\epsilon}(t_{\alpha,n}-t'_{\alpha,n}),\\
\xi_R(t-t')&=&-\sum_{n\geq 0} \sum_ {\beta=-M+1}^{0}\frac{\Gamma (
2+ \frac{\beta}{M})}{\Gamma (n+2 + \frac{\beta}{M} )}
\frac{\lambda^{-M({n+1+\frac{\beta}{M}})}}{\epsilon}(t_{\beta,n}-t'_{\beta,n}).
\end{eqnarray*}
To better understand these identities, following special cases are given explicitly.

Similar to \cite{Fay-like}, we can choose other cases in different
values of $m, t, t'$ which lead to the following Fay-like
identities:\\

{\bf  \Rmnum{1}.} $m=0, t'=t-[\lambda_1^{-1}]^N-[\lambda_2^{-1}]^N$.
In this case the Hirota bilinear identity \eqref{HBE511} will lead to
\begin{eqnarray*}\label{HBE6}
&& \res_{\lambda }
 \left\{ \tau(x,t-[\lambda^{-1}]^N)\times
\tau(x+\epsilon,t'+[\lambda^{-1}]^N)\frac{1}{(1-\lambda
\lambda_1^{-1})(1-\lambda \lambda_2^{-1})}\frac{1}{\lambda}
\right\}  \\
&&=  \res_{\lambda }
 \left\{
\tau(x+\epsilon,t+[\lambda]^M)\times
\tau(x,t'-[\lambda]^M)\frac{1}{\lambda} \right\}.
\end{eqnarray*}
Using \begin{eqnarray*}
(1-\lambda_1^{-1}\lambda)^{-1}(1-\lambda_2^{-1}\lambda)^{-1} =
\frac{\lambda_1^{-1}}{\lambda_1^{-1}-\lambda_2^{-1}}(1-\lambda_1^{-1}\lambda)^{-1}-\frac{\lambda_2^{-1}}{\lambda_1^{-1}-\lambda_2^{-1}}
(1-\lambda_2^{-1}\lambda)^{-1}, \end{eqnarray*} we  get
\begin{eqnarray*}\label{HBE6}
&&  \frac{\lambda_1^{-1}}{\lambda_1^{-1}-\lambda_2^{-1}}
\tau(x,t-[\lambda_1^{-1}]^N) \tau(x+\epsilon,t'+[\lambda_1^{-1}]^N)
-\frac{\lambda_2^{-1}}{\lambda_1^{-1}-\lambda_2^{-1}}
\tau(x,t-[\lambda_2^{-1}]^N)
\tau(x+\epsilon,t'+[\lambda_2^{-1}]^N)  \\
&&= \tau(x+\epsilon,t) \tau(x,t').
\end{eqnarray*}
It further leads to
\begin{eqnarray}\label{HBE61}\notag
&&  \frac{\lambda_1^{-1}}{\lambda_1^{-1}-\lambda_2^{-1}}
\tau(x,t-[\lambda_1^{-1}]^N) \tau(x+\epsilon,t-[\lambda_2^{-1}]^N)
-\frac{\lambda_2^{-1}}{\lambda_1^{-1}-\lambda_2^{-1}}
\tau(x,t-[\lambda_2^{-1}]^N)
\tau(x+\epsilon,t-[\lambda_1^{-1}]^N)  \\
&&= \tau(x+\epsilon,t)
\tau(x,t-[\lambda_1^{-1}]^N-[\lambda_2^{-1}]^N).
\end{eqnarray}

{\bf  \Rmnum{2}.} $m=0, t'=t+[\lambda_1]^M+[\lambda_2]^M$. In this
case the Hirota bilinear identity \eqref{HBE511} will lead to
\begin{eqnarray*}\label{HBE62}
&& \res_{\lambda }
 \left\{ \tau(x,t-[\lambda^{-1}]^N)\times
\tau(x+\epsilon,t'+[\lambda^{-1}]^N)\lambda^{-1}
\right\} \\
&&=  \res_{\lambda }
 \left\{
\tau(x+\epsilon,t+[\lambda]^M)\times
\tau(x,t'-[\lambda]^M)\lambda^{-1}\frac{1}{(1-\lambda^{-1}
\lambda_1)(1-\lambda^{-1} \lambda_2)} \right\}.
\end{eqnarray*}
Using \begin{eqnarray*}
(1-\lambda^{-1}\lambda_1)^{-1}(1-\lambda_2\lambda^{-1})^{-1} =
\frac{\lambda_1}{\lambda_1-\lambda_2}(1-\lambda_1\lambda^{-1})^{-1}-\frac{\lambda_2}{\lambda_1-\lambda_2}
(1-\lambda_2\lambda^{-1})^{-1}, \end{eqnarray*} we  get
\begin{eqnarray*}\label{HBE62}
&&\tau(x,t) \tau(x+\epsilon,t')\\
&=& \frac{\lambda_1}{\lambda_1-\lambda_2}
\tau(x+\epsilon,t+[\lambda_1]^M) \tau(x,t'-[\lambda_1]^M)
-\frac{\lambda_2}{\lambda_1-\lambda_2}
\tau(x+\epsilon,t+[\lambda_2]^M) \tau(x,t'-[\lambda_2]^M).
\end{eqnarray*}
It further leads to
\begin{eqnarray}\label{HBE62'}
&&\tau(x,t) \tau(x+\epsilon,t+[\lambda_1]^M+[\lambda_2]^M)\\\notag
&=& \frac{\lambda_1}{\lambda_1-\lambda_2}
\tau(x+\epsilon,t+[\lambda_1]^M) \tau(x,t+[\lambda_2]^M)
-\frac{\lambda_2}{\lambda_1-\lambda_2}
\tau(x+\epsilon,t+[\lambda_2]^M) \tau(x,t+[\lambda_1]^M).
\end{eqnarray}

{\bf  \Rmnum{3}.} $m=1, t'=t-[\lambda_1^{-1}]^N+[\lambda_2]^M$. In
this case the Hirota bilinear identity \eqref{HBE511} will lead to
\begin{eqnarray*}\label{HBE62}
&& \res_{\lambda }
 \left\{ \tau(x,t-[\lambda^{-1}]^N)\times
\tau(x,t'+[\lambda^{-1}]^N)\frac{1}{1-\lambda \lambda_1^{-1}}
\right\} \\
&&=  \res_{\lambda }
 \left\{
\tau(x+\epsilon,t+[\lambda]^M)\times
\tau(x-\epsilon,t'-[\lambda]^M)\frac{1}{1-\lambda^{-1} \lambda_2}
\right\} ,
\end{eqnarray*}
which is equivalent to
\begin{eqnarray*}\label{HBE62'''}
 \lambda_1 (\tau(x,t-[\lambda_1^{-1}]^N) \tau(x,t'+[\lambda_1^{-1}]^N)-\tau(x,t) \tau(x,t')) =
\lambda_2\tau(x+\epsilon,t+[\lambda_2]^M)
\tau(x-\epsilon,t'-[\lambda_2]^M).
\end{eqnarray*}
It further implies
\begin{eqnarray}\notag
 \lambda_1 (\tau(x,t-[\lambda_1^{-1}]^N)\tau(x,t+[\lambda_2]^M)-\tau(x,t) \tau(x,t-[\lambda_1^{-1}]^N+[\lambda_2]^M))
 \\\label{HBE62''}
=\lambda_2\tau(x+\epsilon,t+[\lambda_2]^M)
\tau(x-\epsilon,t-[\lambda_1^{-1}]^N).
\end{eqnarray}
 These identities  can
be used to prove  the Adler-Shiota-van Moerbeke (ASvM) formula. We tried that but got stuck by some difficulty.
We will omit it because the center of our consideration in this
paper is the HBEs of the EBTH which will appear in the next section.\\

As the end of this section, we would like to show the close
relations between tau function and dynamical functions $w_i$ and
$\tilde w_i$ from Sato equation. Calculate the residue of
eq.\eqref{bn1}, it implies
\begin{eqnarray}
\label{residue}
 \d_{\alpha,n}w_1 & =- \res_{\Lambda}\B_{\alpha,n},
\end{eqnarray}
where the residue is the coefficient of  term $\Lambda^{-1}$.
According to eq.\eqref{pl-1tau}, we have
\begin{equation}
P_L= 1+\frac{w_1}{\lambda}+\frac{w_2}{\lambda^2}+\cdots=1-\frac{\epsilon\partial_{N,0}\tau(x,t)}{\lambda \tau(x,t)}
+ \cdots,
\end{equation}
which implies $w_1= -\epsilon\partial_{N,0}\log\tau(x,t)$. Taking $w_1$ into eq.\eqref{residue}, we have
\begin{eqnarray}
 \epsilon\partial_{\alpha,n}\partial_{N,0}\log\tau(x,t)& = \res_{\Lambda}\B_{\alpha,n},
\end{eqnarray}
Additionally, comparing the coefficient of $\Lambda^0$ on both sides of
eq.\eqref{bn1'}, we can get
\begin{eqnarray*}
\d_{\alpha,n}\tilde w_0 = \tilde w_0
\res_{\Lambda}[(B_{\alpha,n})\Lambda^{-1}].
\end{eqnarray*}
Further considering eq.\eqref{w0tau} and replacing  $x-\frac{\epsilon}{2}$ by $x$, we can get
\begin{eqnarray}
\d_{\alpha,n}\log \tilde w_0 =\d_{\alpha,n}\log
\frac{\tau(x+\epsilon)}{\tau(x)}=
\res_{\Lambda}[(B_{\alpha,n})\Lambda^{-1}].
\end{eqnarray}
The relations between tau function and other dynamical functions
also can be given  by tedious calculation from Sato equations.

\sectionnew{the HBEs of the EBTH}
In this section we continue to  discuss the fundamental properties
of the tau function, i.e., the Hirota bilinear equations. So we
introduce the following vertex operators
\begin{eqnarray*}
\Gamma^{\pm a} :&=&\exp\left({\pm  \sum_{n\geq 0}\left[\sum_
{\alpha=1}^{N} \frac{\Gamma ( 2- \frac{\alpha-1}{N})}{\Gamma (n+2 -
\frac{\alpha-1}{N} )}
\frac{\lambda^{N({n+1-\frac{\alpha-1}{N}})}}{\epsilon}t_{\alpha, n}
+\frac{\lambda^{nN}}{n!}\( \log {\lambda} - \frac{1}{2}(\frac{1}{M}+\frac{1}{N})\C_n\)
\frac{t_{-M,n}}{\epsilon}\] }\right)\\
&&\times\exp\left({\mp
\frac{\epsilon}{2}\d_{-M,0}  \mp \[\lambda^{-1}\]^N_\d  }\right),\\
\Gamma^{\pm b} :
 &=&\exp\left({ \pm \sum_{n\geq 0} \left[\sum_
{\beta=-M+1}^{0}\frac{\Gamma ( 2+ \frac{\beta}{M})}{\Gamma (n+2 +
\frac{\beta}{M} )}
\frac{\lambda^{-M({n+1+\frac{\beta}{M}})}}{\epsilon}t_{\beta,
n}-\frac{\lambda^{-nM}}{n!}\(\log {\lambda}+\frac{1}{2}(\frac{1}{M}+\frac{1}{N})\C_n\)
\frac{t_{-M,n}}{\epsilon}\right] }\right) \\
&&\times \exp\left({\mp \frac{\epsilon}{2}\d_{-M,0} \mp
\[\lambda\]^M_\d }\right),
\end{eqnarray*}

where\ \
 $$\[\lambda^{-1}\]^N_\d=\sum_{n\geq 0}\sum_
{\alpha=1}^{N}\frac{\Gamma (n+1 - \frac{\alpha-1}{N} )}{N\Gamma (
2- \frac{\alpha-1}{N})}
 \epsilon
{\lambda^{-N(n+1-\frac{{\alpha-1}}{N})}
 \frac{\partial}{\partial{
t_{\alpha ,n}}}},$$
$$\[\lambda\]^M_\d=\sum_{n\geq 0}\sum_ {\beta=-M+1}^{0}\frac{\Gamma (n+1
+ \frac{\beta}{M} )}{M\Gamma ( 2+ \frac{\beta}{M})}
 \epsilon
\lambda^{M(n+1+\frac{{\beta}}{M})}
 \frac{\partial}{\partial
t_{\beta ,n}}.$$ We can see that the coefficients of the vertex
operators  $\Gamma^{\pm a} \otimes \Gamma^{\mp a}$ and
$\Gamma^{\pm b} \otimes \Gamma^{\mp b }$  are multi-valued
function because of the logarithmic terms $\log \lambda$. There
are monodromy factors $M^a$ and $M^b$ respectively as following
between two different ones in adjacent
 branches around $\lambda=\infty$
\begin{equation} M^{a}= \exp \left\{ \pm \frac{2\pi i}{\epsilon} \sum_{n\geq 0}
\frac{\lambda^{nN}}{n!} ( t_{-M,n} \otimes 1 - 1\otimes t_{-M,n})
\right\},\end{equation}
\begin{equation}
 M^{b}= \exp \left\{ \pm \frac{2\pi i}{\epsilon}
\sum_{n\geq 0} \frac{\lambda^{-nM}}{n!} ( t_{-M,n} \otimes 1 -
1\otimes t_{-M,n}) \right\}.
\end{equation}
In order to offset the complication we need to generalize the
concept of vertex operators which leads it to be not scalar-valued
any more. So we introduce the following vertex operators
\begin{equation}\Gamma^{\delta}_{a} = \exp\( -\sum_{n>0}\frac{\lambda^{nN}}{\epsilon
n!}(\epsilon\d_x)t_{-M,n}\) \exp(x\d_{-M,0}),\end{equation}
\begin{equation}\Gamma^{\delta}_{b} = \exp\( -\sum_{n>0}\frac{\lambda^{-nM}}{\epsilon
n!}(\epsilon\d_x)t_{-M,n}\) \exp(x\d_{-M,0}).\end{equation}
 Then \begin{equation}
 \label{double delta a} \Gamma^{\delta
\#}_{a}\otimes \Gamma^{\delta}_{a} = \exp(x\d_{-M,0})\exp\(
\sum_{n>0}\frac{\lambda^{nN}}{\epsilon
n!}(\epsilon\d_x)(t_{-M,n}-t'_{-M,n}) \) \exp(x\d'_{-M,0}),
\end{equation}
\begin{equation}
\label{double delta b} \ \ \ \Gamma^{\delta \#}_{b}\otimes
\Gamma^{\delta}_{b} = \exp(x\d_{-M,0})\exp\(
\sum_{n>0}\frac{\lambda^{-nM}}{\epsilon
n!}(\epsilon\d_x)(t_{-M,n}-t'_{-M,n}) \) \exp(x\d'_{-M,0}).
\end{equation}
After computation we get
\begin{eqnarray*} && \(\Gamma^{\delta \#}_{a}\otimes \Gamma^\delta_{a} \) M^{a} =
\exp \left\{ \pm \frac{2\pi i}{\epsilon} \sum_{n> 0}
\frac{\lambda^{nN}}{n!} ( t_{-M,n} - t'_{-M,n})
\right\}\\
&& \exp\(  \pm \frac{2\pi i}{\epsilon} (( t_{-M,0}+x) -(
t'_{-M,0}+x+ \sum_{n> 0} \frac{\lambda^{nN}}{n!} ( t_{-M,n} -
t'_{-M,n})) \) \(\Gamma^{\delta \#}_{a}\otimes \Gamma^\delta _{a}\)
\\&=& \exp\({\pm \frac{2\pi i}{\epsilon}(t_{-M,0}-t'_{-M,0})}\)
\(\Gamma^{\delta \#}_{a}\otimes \Gamma^\delta _{a}\),
\end{eqnarray*}
\begin{eqnarray*}
&& \(\Gamma^{\delta \#}_{b}\otimes \Gamma^\delta_{ b} \) M^{b} =
\exp \left\{ \pm \frac{2\pi i}{\epsilon} \sum_{n> 0}
\frac{\lambda^{-nM}}{n!} ( t_{-M,n} - t'_{-M,n})
\right\}\\
&& \exp\(  \pm \frac{2\pi i}{\epsilon} (( t_{-M,0}+x) -(
t'_{-M,0}+x+ \sum_{n> 0} \frac{\lambda^{-nM}}{n!} ( t_{-M,n} -
t'_{-M,n})) \)\(\Gamma^{\delta \#}_{b}\otimes \Gamma^\delta_{b} \)
\\&=& \exp\({\pm \frac{2\pi i}{\epsilon}(t_{-M,0}-t'_{-M,0})}\)
\(\Gamma^{\delta \#}_{b}\otimes \Gamma^\delta_{b} \).
\end{eqnarray*}
Thus when $t_{-M,0}-t'_{-M,0} \in \Z\epsilon $, $\(\Gamma^{\delta
\#}_{a}\otimes \Gamma^{\delta}_{a}\) \( \Gamma^{a}\otimes
\Gamma^{-a}\) \mbox{and}\(\Gamma^{\delta \#}_{b}\otimes
\Gamma^{\delta}_{b }\)\(\Gamma^{-b}\otimes\Gamma^{b}\)$ are all
single-valued near $\lambda=\infty$.

 We will say that $\tau$ satisfies the {\bf HBEs of the EBTH}
 if
\begin{equation} \label{HBE} \res_{{\rm{\lambda}}}
 \left(\lambda^{Nr-1}\(\Gamma^{\delta
\#}_{a}\otimes \Gamma^{\delta}_{a}\) \( \Gamma^{a}\otimes
\Gamma^{-a}\right) -\lambda^{-Mr-1}\(\Gamma^{\delta \#}_{b}\otimes
\Gamma^{\delta}_{b }\)\(\Gamma^{-b}\otimes\Gamma^{b} \) \right)(\tau
\otimes \tau )=0
\end{equation}
computed at $t_{-M,0}-t'_{-M,0}=m\epsilon$
 for each  $m\in \Z$, $r\in \N$.
 Now we should note that the vertex operators  take
 value in algebra $A[[t]]$ whose element is like
 $\sum_{i\geq 0}c_i(x,t,\epsilon)\d^i$.

\begin{theorem}\label{t11}
Function $\tau(t,\epsilon)$  is a tau-function of the extended
bigraded Toda hierarchy at a certain spatial point if and only if it
satisfies the Hirota bilinear equations \eqref{HBE}.
\end{theorem}
{\bf Proof.} Note that the  $\tau(t)$ now is independent of
variable $x$ because the $x$ takes a fixed value, fox example
$x=x_0$(constant). However, in the following proof, $x$ will appear
in the $\tau(t)$  due to the action of vertex operator on $\tau(t)$. For
example, $e^{x\partial_{-M,0}}\tau(t)$ = $\tau(t_{-M,0}+x,\bar{t})$
where  $\bar{t}$ is  just as the definition in the proof of the existence of tau function.

 We just need  to prove that the HBEs
are equivalent to the right side in Proposition
\ref{wave-operators}. By a straightforward computation we can get
the following four identities {\allowdisplaybreaks}
\begin{eqnarray}\label{vertex computation1}
\Gamma^{\delta \#}_a \Gamma^{a}\tau & =& \tau(t_{-M,0}+x-\epsilon/2,\bar{t})
\lambda^{\ t_{-M,0}/\epsilon} W_L(x,t,\epsilon \d_x,\lambda )\lambda^{x/\epsilon},
\\ \label{vertex computation2}
 \Gamma^{\delta}_a \Gamma^{-a}\tau  & =&
\lambda^{-t_{-M,0}/\epsilon} \lambda^{ -x/\epsilon}
W_L^{-1}(x,t,\epsilon \d_x,\lambda)\ \tau(x+t_{-M,0}+\epsilon/2,\bar{t}), \\\label{vertex
computation3}
 \Gamma^{\delta \#}_b \Gamma^{-b}\tau  & =&
\tau(x+t_{-M,0}-\epsilon/2,\bar{t}) \lambda^{t_{-M,0}/\epsilon} W_R(x,t,\epsilon \d_x,\lambda
)\lambda^{ x/\epsilon}, \\\label{vertex computation4}
 \Gamma^{\delta}_b \Gamma^{b}\tau & = &\lambda^{-t_{-M,0}/\epsilon} \lambda^{
 -x/\epsilon} W_R^{-1}(x,t,\epsilon \d_x,\lambda)\
\tau(x+t_{-M,0}+\epsilon/2,\bar{t}) .
\end{eqnarray}
Here $\bar t$ is denoted as all the other time variables except $t_{-M,0}$.  We should note that we take the left side of
eq.\eqref{vertex computation1}-eq.\eqref{vertex computation4} not as
functions but operators involving $e^{\partial_x}$. We should pay
more attention to the different operations of the operators $\d_x$,
$\partial_{\alpha,n}$ and $\partial_{\beta,m}$, for example,
\begin{eqnarray}\label{translation1}
&&\exp\(\sum_{n>0}\frac{\lambda^{nN}}{
n!}t_{-M,n}\d_x \)\tau(x+t_{_M,0},\bar{t},\lambda)\nonumber \\
&&=\tau(x+t_{-M,0}+\sum_{n>0}\frac{\lambda^{nN}}{ n!}
 t_{-M,n}, \bar{t},\lambda)
\exp\(\sum_{n>0}\frac{\lambda^{nN}}{
n!}t_{-M,n}\d_x \),
\end{eqnarray}
\begin{eqnarray}\label{translation2}
 \exp\(x\d_{-M,0}\) \exp \left\{-\frac{\epsilon}{2}\d_{-M,0}  -
[\lambda^{-1}]^N_\d \right\} \tau(t;\epsilon)= \tau(t_{-M,0}+x-\epsilon/2,\bar t-[\lambda^{-1}]^N),
\end{eqnarray}
\begin{eqnarray}\label{translation3}
&& \exp\left\{(\log \lambda ) \frac{t_{-M,0}+x}{\epsilon}
 \right\}
\exp\(\sum_{n>0}\frac{\lambda^{nN}}{\epsilon\
n!}(\epsilon\d_x)t_{-M,n}\)  \nonumber \\
 &=&\exp\(\sum_{n>0}\frac{\lambda^{nN}}{\epsilon\ n!}
(\epsilon\d_x)t_{-M,n}\)
 \exp\left\{(\log \lambda )\frac{t_{-M,0}+x-\sum\limits_{n>0}\dfrac{\lambda^{nN}}{\epsilon n!}\epsilon t_{-M,n}  }{\epsilon}
 \right\}.
\end{eqnarray}
The above formula shows the relationship between $e^{\partial_x}$ and $\tau(t)$ is a product of operators, but
the relationship between $e^{\partial_{\alpha,n}}$ (or $e^{\partial_{\beta,m}}$) and $\tau(t)$  is a action of the
former on the latter.

For simplifying the proof, we first introduce following operators,
\begin{equation}\label{D}
D=\sum\limits_{n\geq 0}\sum\limits_{\alpha=1}^{N}\dfrac{\Gamma(2-\frac{\alpha-1}{N})}{
\Gamma(n+2-\frac{\alpha-1}{N})} \dfrac{\lambda^{N(n+1-\frac{\alpha-1}{N})}}{\epsilon}t_{\alpha,n},
\end{equation}
\begin{equation}\label{E}
E=\sum\limits_{n> 0}\dfrac{\lambda^{nN}}{n!}(\log \lambda-\frac{1}{2}(\frac{1}{M}+\frac{1}{N})\C_n)\frac{t_{-M,n}}{\epsilon},
\end{equation}
then, with the help of above identities eq.\eqref{translation1} and eq.\eqref{translation2}, the left hand side of identity eq.\eqref{vertex computation1} can be expressed by
{\allowdisplaybreaks
\begin{eqnarray*}
\Gamma^{\delta \#}_a \Gamma^{a}\tau &=&
 \exp\(x\d_{-M,0}\)
\exp\(\sum_{n>0}\frac{\lambda^{nN}}{\epsilon
n!}(\epsilon\d_x)t_{-M,n}\)\times \\
&& \exp \left\{\sum_{n\geq 0} \left[\sum_
{\alpha=1}^{N}\frac{\Gamma ( 2- \frac{\alpha-1}{N})}{\Gamma (n+2 -
\frac{\alpha-1}{N} )}
\frac{\lambda^{N({n+1-\frac{\alpha-1}{N}})}}{\epsilon}t_{\alpha,
n} +\frac{\lambda^{nN}}{n!}\( \log {\lambda} - \frac{1}{2}(\frac{1}{M}+\frac{1}{N})\C_n\)
\frac{t_{-M,n}}{\epsilon}\right]  \right\}\\
&& \times \exp \left\{-\frac{\epsilon}{2}\d_{-M,0}  -
[\lambda^{-1}]^N_\d \right\} \tau(t;\epsilon)\\
&&=\exp\{D\}\exp\{E \} \exp\{x\partial_{-M,0} \}\exp\{\sum\limits_{n>0} \dfrac{\lambda^{nN}}{\epsilon\ n!}
(\epsilon \partial_x) t_{-M,n} \}\\
&&
\exp\{(\log\lambda)\frac{t_{-M,0}}{\epsilon} \} \tau(t_{-M,0}-\frac{\epsilon}{2},\bar{t}-[\lambda^{-1}]^N )\\
&&=\tau(t_{-M,0}+x-\frac{\epsilon}{2},\bar{t}-[\lambda^{-1}]^N )\exp\{D\}\exp\{E \}
\exp\{(\log\lambda)\frac{t_{-M,0}+x}{\epsilon} \}\\
&&
\exp\{\sum\limits_{n>0} \dfrac{\lambda^{nN}}{\epsilon\ n!}
(\epsilon \partial_x) t_{-M,n} \}.
\end{eqnarray*}
Taking eq.\eqref{translation3} into it,  then substituting $D$ and $E$ by eq.\eqref{D} and eq.\eqref{E},

\begin{eqnarray*}
\Gamma^{\delta \#}_a \Gamma^{a}\tau
&=&\tau(t_{-M,0}+x-\epsilon/2,\bar t)P_L(x,t,\lambda)
\exp\(\sum_{n>0}\frac{\lambda^{nN}}{\epsilon n!}
(\epsilon\d_x)t_{-M,n}\) \\
&& \exp \left\{ \left[\sum_{n\geq 0}
 \sum_
{\alpha=1}^{N}\frac{\Gamma ( 2- \frac{\alpha-1}{N})}{\Gamma (n+2 -
\frac{\alpha-1}{N} )}
\frac{\lambda^{N({n+1-\frac{\alpha-1}{N}})}}{\epsilon}t_{\alpha, n}
+ \sum_{n>0}
\frac{\lambda^{nN}}{n!}\( \log {\lambda} - \frac{1}{2}(\frac{1}{M}+\frac{1}{N})\C_n\) \frac{t_{-M,n}}{\epsilon}\right] \right\}\\
&&\exp\left\{(\log \lambda)
\[t_{-M,0}-\(\sum_{n>0}\frac{\lambda^{nN}}{n!}t_{-M,n}\)+x\]/\epsilon
 \right\}\\
 &=&
\tau(t_{-M,0}+x-\epsilon/2,\bar t)P_L(x,t,\lambda) \times \\&& \exp \left\{
\left[\sum_{n\geq 0}
 \sum_
{\alpha=1}^{N}\frac{\Gamma ( 2- \frac{\alpha-1}{N})}{\Gamma (n+2 -
\frac{\alpha-1}{N} )}
\frac{\lambda^{N({n+1-\frac{\alpha-1}{N}})}}{\epsilon}t_{\alpha,
n} + \sum_{n>0}\frac{ \lambda^{nN}}{n!}
(\epsilon\d_x-\frac{1}{2}(\frac{1}{M}+\frac{1}{N})\C_n)\frac{t_{-M,n}}{\epsilon}
\right]\right\}\\
&&\exp\left\{(\log \lambda )
(t_{-M,0}+x)/\epsilon
 \right\}\\
 & =&
\tau(t_{-M,0}+x-\epsilon/2,\bar t) \lambda^{\ t_{-M,0}/\epsilon} W_L(x,t,\epsilon \d_x,\lambda
)\lambda^{x/\epsilon}.
 \end{eqnarray*}
 The other
three identities are derived in
 similar way which will be shown in detail in the appendix. \\
By substituting four equations eq.\eqref{vertex
computation1}-eq.\eqref{vertex computation4} into the HBEs
\eqref{HBE} we find:
\begin{eqnarray*}
&&\res_{\lambda }
 \left\{
\lambda^{Nr-1}\Gamma^{\delta\#}\Gamma^{a}\tau \otimes
\Gamma^{\delta}\Gamma^{-a}\tau -\lambda^{-Mr-1}
\Gamma^{\delta\#}\Gamma^{-b}\tau \otimes
\Gamma^{\delta}\Gamma^{b}\tau \right\} \\
= &&\res_{\lambda }
 \left\{ \tau(x-\epsilon/2,t)
\lambda^{Nr-1}\lambda^{(t_{-M,0}-t'_{-M,0})/\epsilon}
W_L(x,t,\epsilon \d_x,\lambda) W_L^{-1}(x,t',\epsilon \d_x,\lambda)\
\tau(x+\epsilon/2,t') \right. \\
 && \left.  - \tau(x-\epsilon/2,t)
\lambda^{-Mr-1}\lambda^{(t_{-M,0}-t'_{-M,0})/\epsilon}
W_R(x,t,\epsilon \d_x,\lambda) W_R^{-1}(x,t',\epsilon \d_x,\lambda)\tau(x+\epsilon/2,t')
\right\}.
\end{eqnarray*}
Note here $\tau(x-\frac{\epsilon}{2}, t)=\tau(t_{-M,0}+x-\frac{\epsilon}{2},\bar t)$ as eq.(\ref{eq34'''}).
Let $t_{-M,0}-t'_{-M,0} = m\epsilon$ and
consider that $W_L(x,t,\epsilon \d_x,\lambda), W_L^{-1}(x,t',\epsilon \d_x,\lambda)$ and
$W_R(x,t,\epsilon \d_x,\lambda), W_R^{-1}(x,t',\epsilon \d_x,\lambda)$ are all not scaled-valued
but take values in the algebra of differential operator.  Therefore the HBEs lead to
\begin{eqnarray*}
&& \res_{\lambda }
 \left\{\lambda^{m+Nr-1}
W_L(x,t_{-M,0},\bar t,\epsilon \d_x,\lambda)
W_L^{-1}(x,t_{-M,0}-m\epsilon,\bar t',\epsilon \d_x,\lambda)-\right.\\
&& \left.\lambda^{m-Mr-1} W_R(x,t_{-M,0},\bar t,\epsilon \d_x,\lambda) W_R^{-1}
(x,t_{-M,0}-m\epsilon,\bar t',\epsilon \d_x,\lambda) \right\}=0,
\end{eqnarray*}
which can also be written as
\begin{eqnarray*} && \res_{\lambda }
\left\{\lambda^{m+Nr-1} W_L(x,t_{-M,0}\bar t,\epsilon \d_x,\lambda)
W_L^{-1}(x-m\epsilon,t_{-M,0},\bar t',\epsilon \d_x,\lambda)-\right. \\
&& \left.\lambda^{m-Mr-1} W_R(x,t_{-M,0},\bar t,\epsilon \d_x,\lambda) W_R^{-1}
(x-m\epsilon,t_{-M,0},\bar t',\epsilon \d_x,\lambda) \right\}=0.
\end{eqnarray*}
 This is just
eq.\eqref{HBE3}. So the proof is finished.\qed

\sectionnew{the HBEs of BTH}
Excluding the variables of $t_{-M,n},  n\geq 1$, we obtain a Hirota bilinear
equations for the  bigraded Toda  hierarchy (BTH).
Similar to  EBTH, we introduce the following vertex operators
\begin{eqnarray*} \Gamma^{\pm c} && = \exp \left\{ \pm \left[\sum_
{\alpha=1}^{N}\frac{\Gamma ( 2- \frac{\alpha-1}{N})}{\Gamma (n+2 -
\frac{\alpha-1}{N} )}
\frac{\lambda^{N({n+1-\frac{\alpha-1}{N}})}}{\epsilon}t_{\alpha,
n}
\right]  \right\}\\
&& \times \exp \left\{\mp \frac{\epsilon}{2}\d_{-M,0}  \mp
\sum_{n\geq 0}\left[\sum_ {\alpha=1}^{N}\frac{\Gamma (n+1 -
\frac{\alpha-1}{N} )}{N\Gamma ( 2- \frac{\alpha-1}{N})}
\lambda^{-N(n+1-\frac{{\alpha-1}}{N})}
 \frac{\partial}{\partial{
t_{\alpha ,n}}}\right]  \right\},
\end{eqnarray*}

\begin{eqnarray*} \Gamma^{\pm d} && = \exp \left\{
\pm \left[\sum_ {\beta=-M+1}^{0}\frac{\Gamma ( 2+
\frac{\beta}{M})}{\Gamma (n+2 +\frac{\beta}{M} )}
\frac{\lambda^{-M({n+1+\frac{\beta}{M}})}}{\epsilon}t_{\beta, n}
\right]  \right\}\\
&& \times \exp \left\{\mp \frac{\epsilon}{2}\d_{-M,0}  \mp
\sum_{n\geq 0}\left[\sum_ {\beta=-M+1}^{0}\frac{\Gamma (n+1 +
\frac{\beta}{M} )}{M\Gamma ( 2+ \frac{\beta}{M})}
\lambda^{M(n+1+\frac{{\beta}}{M})}
 \frac{\partial}{\partial{
t_{\beta ,n}}}\right]  \right\} . \end{eqnarray*} In this case,
because there is no logarithmic term in the vertex operators, so we
need not generalize the vertex operators. Just as a result of that,
the vertex operator will take values in scalared function of
$\epsilon, t, \lambda$.
\begin{corollary}\label{c1} A non-vanishing function
$\tau(t_{-M,0};t_{\alpha,n},\dots;\epsilon)$ is a tau-function of
the  BTH ($t_{-M,n},  n\geq 1$ excluded) if and only if for
each  $m\in \Z$, $r\in \N$,
\begin{eqnarray} \res\limits_{{\rm{\lambda}}} \left\{
\lambda^{Nr+m-1} \Gamma^{c}\otimes \Gamma^{-c} - \lambda^{-Mr+m-1}
\Gamma^{-d}\otimes \Gamma^{d} \right\} (\tau \otimes \tau)=0,
\end{eqnarray} \indent when $t_{-M,0}-t'_{-M,0}=m\epsilon$.
\end{corollary}

\sectionnew{Conclusions and Discussions}
In previous sections, we have succeeded in extending  Sato theory to
the EBTH. Starting from the revised definition of the  Lax
equations,  we have given Sato equations, wave operators, Hirota
bilinear identities related to the wave operators, the existence of
the tau function  and its  important properties including Fay-like
identities and Hirota bilinear equations.  In particular, this
hierarchy deserves further studying and exploring because of its
potential applications in topological quantum fields  and
Gromov-Witten theory. Our main support of this statement currently
is that ETH describes the Gromov-Witten invariants of $CP^1$. We
would like to point out that our Lax equations are revised from
Carlet's result \cite{C}, but our proof on the existence of the tau functions
is more transparent than it.

 Our future work will contain the applications of this kind of HBEs in the
topological fields theory and string theory, the virasoro constraint
of EBTH from the point of string equation and ASvM formula.

{\bf {Acknowledgements:}}
  {\small    This work is supported by the NSF of China under Grant No. 10671187. It is also supported by
Program for NCET under Grant No. NECT-08-0515. We are  thankful to Professor Youjin Zhang, Dr. Siqi Liu (Tsinghua University,
    China), Prof. Weizhong Zhao (CNU, China)  for helpful discussions and Dr. Todor E. Milanov
    (North Carolina State University, USA) for  his valuable suggestions by email. The author Chuanzhong Li
also thanks Professor Yuji Kodama and Professor Hsian-Hua Tseng in Ohio State University for their  useful discussion  in their colloquium.
 We also thank Professor Li Yishen
(USTC, China) for long-term encouragements and supports.  We thank anonymous referee for their
valuable suggestions and pertinent criticisms.}
%%%%%%%%%%%%%%%%% References  %%%%%%%%%%%%%%%%%%%%%%%%%%%%%%%%%%%%%%%
\newpage{}

\sectionnew{Appendix}

 {\bf Proof of the identities eq.\eqref{vertex
computation2}-eq.\eqref{vertex computation4}}:

By a similar calculation of eq.\eqref{vertex computation1}, we have
 {\allowdisplaybreaks
 \begin{eqnarray*}
&&\Gamma^{\delta }_a \Gamma^{-a}\tau\\
 &=&
\exp\(-\sum_{n>0}\frac{\lambda^{nN}}{\epsilon
n!}(\epsilon\d_x)t_{-M,n}\)\times  \exp\(x\d_{-M,0}\)\\
&& \exp \left\{-\sum_{n\geq 0} \left[\sum_
{\alpha=1}^{N}\frac{\Gamma ( 2- \frac{\alpha-1}{N})}{\Gamma (n+2 -
\frac{\alpha-1}{N} )}
\frac{\lambda^{N({n+1-\frac{\alpha-1}{N}})}}{\epsilon}t_{\alpha, n}
+\frac{\lambda^{nN}}{n!}\( \log {\lambda} - \frac{1}{2}(\frac{1}{M}+\frac{1}{N})\C_n\)
\frac{t_{-M,n}}{\epsilon}\right]  \right\}\\
 &&
 \times \exp \left\{\frac{\epsilon}{2}\d_{-M,0}+[\lambda^{-1}]^N_\d \right\}
 \tau(t)\\
&=& \exp \left\{\left[-\sum_{n\geq 0}
 \sum_{\alpha=1}^{N}\frac{\Gamma ( 2- \frac{\alpha-1}{N})}{\Gamma
(n+2 - \frac{\alpha-1}{N} )}
\frac{\lambda^{N({n+1-\frac{\alpha-1}{N}})}}{\epsilon}t_{\alpha, n}
-\sum_{n>0} \frac{\lambda^{nN}}{n!}\( \epsilon\d_x -
 \frac{1}{2}(\frac{1}{M}+\frac{1}{N})\C_n\)
\frac{t_{-M,n}}{\epsilon}\right] \right\}\\
&& \exp\(-\sum_{n>0}\frac{\lambda^{nN}}{\epsilon n!}
\frac{t_{-M,n}}{\epsilon}\log {\lambda}\)\exp\left\{-(\log \lambda
)(t_{-M,0}+x)/\epsilon
 \right\}\tau(x+t_{-M,0}+\epsilon/2,\bar{t}+[\lambda^{-1}]^N)\\
& =&\lambda ^{-\frac{t_{-M,0}+x}{\epsilon}} \exp
\left\{\left[-\sum_{n\geq 0}
 \sum_{\alpha=1}^{N}\frac{\Gamma ( 2- \frac{\alpha-1}{N})}{\Gamma
(n+2 - \frac{\alpha-1}{N} )}
\frac{\lambda^{N({n+1-\frac{\alpha-1}{N}})}}{\epsilon}t_{\alpha, n}
-\sum_{n>0} \frac{\lambda^{nN}}{n!}\( \epsilon\d_x - \frac{1}{2}(\frac{1}{M}+\frac{1}{N})\C_n\)
\right.\right.\\&&
\left.\left.
\frac{t_{-M,n}}{\epsilon}\right] \right\}\tau(x+t_{-M,0}+\epsilon/2,\bar{t}+[\lambda^{-1}]^N)\\
&=&\lambda ^{-\frac{t_{-M,0}+x}{\epsilon}} \exp
\left\{\left[-\sum_{n\geq 0}
 \sum_{\alpha=1}^{N}\frac{\Gamma ( 2- \frac{\alpha-1}{N})}{\Gamma
(n+2 - \frac{\alpha-1}{N} )}
\frac{\lambda^{N({n+1-\frac{\alpha-1}{N}})}}{\epsilon}t_{\alpha, n}
-\sum_{n>0} \frac{\lambda^{nN}}{n!}\( \epsilon\d_x - \frac{1}{2}(\frac{1}{M}+\frac{1}{N})\C_n\)\right.\right.\\&&
\left.\left.
\frac{t_{-M,n}}{\epsilon}\right] \right\}
P_L^{-1}(x,t,\lambda)\tau(x+t_{-M,0}+\epsilon/2,\bar{t})\\
&=&
 \lambda ^{-\frac{t_{-M,0}+x}{\epsilon}}
W_L^{-1}(x,t,\epsilon\d_x,\lambda )\tau(x+t_{-M,0}+\epsilon/2,\bar{t}).
\end{eqnarray*}
So eq.\eqref{vertex computation2} is proved.

For the convenience of the proof of eq.\eqref{vertex computation3},
we introduce an identity

\begin{eqnarray}\label{translation4}
&&\exp\{x\partial_{-M,0} \}
\exp\{\sum\limits_{n>0}\dfrac{\lambda^{-nM}}{\epsilon\ n!} (\epsilon \partial_x)t_{-M,n} \}
\exp\{  (\log \lambda ) \dfrac{t_{-M,0}}{\epsilon} \} \nonumber\\
&&=\exp\{\sum\limits_{n>0}\dfrac{\lambda^{-nM}}{\epsilon\ n!} (\epsilon \partial_x)t_{-M,n} \}
\exp\{  (\log \lambda ) \dfrac{t_{-M,0}+x-\sum\limits_{n>0}\dfrac{\lambda^{-nM}}{\epsilon \ n!} t_{-M,n}}{\epsilon} \},
\end{eqnarray}

and define two operators
\begin{eqnarray}
&&F=\sum_{n\geq 0}
\left[\sum_ {\beta=-M+1}^{0}\frac{\Gamma ( 2+
\frac{\beta}{M})}{\Gamma (n+2 +\frac{\beta}{M} )}
\frac{\lambda^{-M({n+1+\frac{\beta}{M}})}}{\epsilon}t_{\beta, n}
\right],  \label{F}\\
&&G=\sum_{n> 0}
\frac{\lambda^{-nM}}{n!}\(- \log {\lambda} - \frac{1}{2}(\frac{1}{M}+\frac{1}{N})\C_n\)
\frac{t_{-M,n}}{\epsilon}   \label{G},
 \end{eqnarray}
then
\begin{eqnarray*}&&\Gamma^{\delta \#}_b \Gamma^{-b}\tau\\
& =&
 \exp\(x\d_{-M,0}\)\times\exp\(\sum_{n>0}\frac{\lambda^{-nM}}{\epsilon
n!}(\epsilon\d_x)t_{-M,n}\) \\
&&
\exp \left\{-\sum_{n\geq 0}
\left[\sum_ {\beta=-M+1}^{0}\frac{\Gamma ( 2+
\frac{\beta}{M})}{\Gamma (n+2 +\frac{\beta}{M} )}
\frac{\lambda^{-M({n+1+\frac{\beta}{M}})}}{\epsilon}t_{\beta, n}
+\frac{\lambda^{-nM}}{n!}\(- \log {\lambda} - \frac{1}{2}(\frac{1}{M}+\frac{1}{N})\C_n\)
\frac{t_{-M,n}}{\epsilon}\right]  \right\} \\
&&
 \times \exp \left\{\frac{\epsilon}{2}\d_{-M,0}+[\lambda]^M_\d \right\} \tau(t;\epsilon) \\
&&=\tau(t_{M,0}+x+\frac{\epsilon}{2},\bar{t}+[\lambda]^M)\exp\{-F\}\exp\{-G\}\\
 && \exp\{x\partial_{-M,0} \}
\exp\{\sum\limits_{n>0}\dfrac{\lambda^{-nM}}{\epsilon\ n!} (\epsilon \partial_x)t_{-M,n} \}
\exp\{  (\log \lambda ) \dfrac{t_{-M,0}}{\epsilon} \}.
\end{eqnarray*}
Using identity eq.\eqref{translation4}, then substituting $F$ and
$G$ given by  eq.\eqref{F} and eq.\eqref{G}, we have
\begin{eqnarray*}
&&\Gamma^{\delta \#}_b \Gamma^{-b}\tau= \tau(x+t_{-M,0}+\epsilon/2,\bar{t}+[\lambda]^M) \\
&&\exp \left\{\left[-\sum_{n\geq 0}
 \sum_{\beta=-M+1}^{0}\frac{\Gamma ( 2+ \frac{\beta}{M})}{\Gamma
(n+2 + \frac{\beta}{M} )}
\frac{\lambda^{-M({n+1+\frac{\beta}{M}})}}{\epsilon}t_{\beta, n}
-\sum_{n>0} \frac{\lambda^{-nM}}{n!}\( -\epsilon\d_x -
 \frac{1}{2}(\frac{1}{M}+\frac{1}{N})\C_n\)
\frac{t_{-M,n}}{\epsilon}\right] \right\} \\
&& \exp\(\sum_{n>0}\frac{\lambda^{-nM}}{\epsilon n!}
\frac{t_{-M,n}}{\epsilon}\log {\lambda}\)\exp\left\{(\log \lambda
)(t_{-M,0}+x-\sum_{n>0}\frac{\lambda^{-nM}}{\epsilon n!} \epsilon
t_{-M,n})/\epsilon
 \right\} \\
&
=&\tau(x+t_{M,0}-\epsilon/2,\bar{t})P_R(x,t,\lambda) \\
&&\exp \left\{-\sum_{n\geq 0}
 \sum_{\beta=-M+1}^{0}\frac{\Gamma ( 2+\frac{\beta}{M})}{\Gamma
(n+2 +\frac{\beta}{M} )}
\frac{\lambda^{-M({n+1+\frac{\beta}{M}})}}{\epsilon}t_{\beta, n}
+\sum_{n>0} \frac{\lambda^{-nM}}{n!}\( \epsilon\d_x +\frac{1}{2}(\frac{1}{M}+\frac{1}{N})\C_n\)
\frac{t_{-M,n}}{\epsilon} \right\}\\
&& \lambda
^{\frac{t_{-M,0}+x}{\epsilon}} \\
& =& \tau(x+t_{M,0}-\epsilon/2,\bar{t})W_R(x,t, \epsilon \d_x
,\lambda )\lambda ^{\frac{t_{-M,0}+x}{\epsilon}},
\end{eqnarray*}
which is eq.\eqref{vertex computation3}.

Similarly, to prove eq.\eqref{vertex computation4}, we have
\begin{eqnarray*}
&&
\Gamma^{\delta }_b \Gamma^{b}\tau \\
&=& \exp\(-\sum_{n>0}\frac{\lambda^{-nM}}{\epsilon
n!}(\epsilon\d_x)t_{-M,n}\)\times  \exp\(x\d_{-M,0}\)\\
&& \exp \left\{\sum_{n\geq 0} \left[\sum_
{\beta=-M+1}^{0}\frac{\Gamma ( 2+ \frac{\beta}{M})}{\Gamma (n+2
+\frac{\beta}{M} )}
\frac{\lambda^{-M({n+1+\frac{\beta}{M}})}}{\epsilon}t_{\beta, n}
+\frac{\lambda^{-nM}}{n!}\( -\log {\lambda} - \frac{1}{2}(\frac{1}{M}+\frac{1}{N})\C_n\)
\frac{t_{-M,n}}{\epsilon}\right]  \right\}\\
&&
 \times \exp \left\{-\frac{\epsilon}{2}\d_{-M,0}-[\lambda]^M_\d \right\} \tau(t)\\
&=& \exp \left\{\left[\sum_{n\geq 0}
 \sum_{\beta=-M+1}^{0}\frac{\Gamma ( 2+ \frac{\beta}{M})}{\Gamma
(n+2 + \frac{\beta}{M} )}
\frac{\lambda^{-M({n+1+\frac{\beta}{M}})}}{\epsilon}t_{\beta, n}
+\sum_{n>0} \frac{\lambda^{-nM}}{n!}\( -\epsilon\d_x -
 \frac{1}{2}(\frac{1}{M}+\frac{1}{N})\C_n\)
\frac{t_{-M,n}}{\epsilon}\right] \right\}\\
&& \exp\(-\sum_{n>0}\frac{\lambda^{-nM}}{\epsilon n!}
\frac{t_{-M,n}}{\epsilon}\log {\lambda}\)\exp\left\{(-\log \lambda
)(t_{-M,0}+x)/\epsilon
 \right\}\tau(x+t_{-M,0}-\epsilon/2,\bar{t}-[\lambda]^M)\\
&=&\lambda ^{-\frac{t_{-M,0}+x}{\epsilon}} \exp \left\{\sum_{n\geq
0}
 \sum_{\beta=-M+1}^{0}\frac{\Gamma ( 2+\frac{\beta}{M})}{\Gamma
(n+2 +\frac{\beta}{M} )}
\frac{\lambda^{-M({n+1+\frac{\beta}{M}})}}{\epsilon}t_{\beta, n}
+\sum_{n>0} \frac{\lambda^{-nM}}{n!}\(- \epsilon\d_x - \frac{1}{2}(\frac{1}{M}+\frac{1}{N})\C_n\)\right.\\&&
\left.
\frac{t_{-M,n}}{\epsilon} \right\} \tau(x+t_{-M,0}-\epsilon/2,\bar{t}-[\lambda^{-1}]^M)\\
&=&\lambda ^{-\frac{t_{-M,0}+x}{\epsilon}} \exp \left\{\sum_{n\geq
0} \sum_{\beta=-M+1}^{0}\frac{\Gamma ( 2+ \frac{\beta}{M})}{\Gamma
(n+2 + \frac{\beta}{N} )}
\frac{\lambda^{-M({n+1+\frac{\beta}{M}})}}{\epsilon}t_{\beta, n} +
\sum_{n> 0}\frac{\lambda^{-nM}}{n!}\( -\epsilon\d_x - \frac{1}{2}(\frac{1}{M}+\frac{1}{N})\C_n\)\right.\\&&
\left.
\frac{t_{-M,n}}{\epsilon} \right\} P_R^{-1}(x,\t,\lambda)\tau(x+t_{-M,0}+\epsilon/2,\bar{t})\\
&=&
 \lambda ^{-\frac{t_{-M,0}+x}{\epsilon}}
W_R^{-1}(x,\t,\epsilon \d_x ,\lambda
)\tau(x+t_{-M,0}+\epsilon/2,\bar{t}).
\end{eqnarray*} \qed

%%%%%%%%%%%%%%%%%%%%%%%%%%%%%%%%%%%%%%%%%%%%%%%%%%%%%%%%%%%%%%%%%%%%%%%%%%%%%%%

%---------------------------------------------------------------------------------------

\end{document}